%% file: partition.tex
\documentclass[USenglish,a4paper,11pt]{article}
\pdfoutput=1

\input{Header/Style.sty}
\title{A Sublinear-Time Quantum Algorithm for \\ Approximating Partition Functions}
\author{Arjan Cornelissen\thanks{QuSoft, University of Amsterdam. \url{arjan.cornelissen@cwi.nl}}~\thanks{IBM Quantum, IBM T.J. Watson Research Center. \url{arjan.cornelissen@ibm.com}}
   \and Yassine Hamoudi\thanks{Simons Institute for the Theory of Computing, University of California, Berkeley. \url{ys.hamoudi@gmail.com}}}
\date{\today}

\begin{document}

\maketitle


\begin{abstract}
  \input{Sources/abstract}
\end{abstract}

\section{Introduction}
\input{Sources/introduction}

\section{Unbiased and nondestructive quantum subroutines}

  \subsection{Phase estimation}
  \label{Sec:unbiasedPE}
  \input{Sources/unbiasedPE}

  \subsection{Amplitude estimation}
  \label{Sec:unbiasedAE}
  \input{Sources/unbiasedAE}

\section{Mean estimation}
\label{Sec:mean}
\input{Sources/mean}

\section{Partition function estimation}
\label{Sec:partition}
  \input{Sources/partition_intro}

  \subsection{General algorithm for partition function estimation}
  \label{Sec:annealing}
  \input{Sources/annealing}

  \subsection{Applications}
  \label{Sec:applications}
  \input{Sources/applications}

\section*{Acknowledgements}
\input{Sources/acknowledgements}

\printbibliography[heading=bibintoc]

\begin{appendices}
  \section{Algorithmic primitives}
  \label{Sec:primitives}
  \input{Sources/app_primitives}

  \section{Nondestructive amplitude estimation}
  \label{App:nondestrAE}
  \input{Sources/app_nondestrAE}

  \section{Volume estimation of convex bodies}
  \label{App:volumeEstimation}
  \input{Sources/volume_estimation}
\end{appendices}

\end{document}

%% file: Sources/abstract.tex
We present a novel quantum algorithm for estimating Gibbs partition functions in \emph{sublinear time} with respect to the logarithm of the size of the state space. This is the first speed-up of this type to be obtained over the seminal nearly-linear time  algorithm of {\v{S}}tefankovi{\v{c}}, Vempala and Vigoda~\cite{SVV09j}. Our result also preserves the quadratic speed-up in precision and spectral gap achieved in previous work by exploiting the properties of quantum Markov chains. As an application, we obtain new polynomial improvements over the best-known algorithms for computing the partition function of the Ising model, counting the number of $k$-colorings, matchings or independent sets of a graph, and estimating the volume of a convex body.

Our approach relies on developing new variants of the quantum phase and amplitude estimation algorithms that return nearly \emph{unbiased} estimates with \emph{low variance} and \emph{without destroying} their initial quantum state. We extend these subroutines into a nearly unbiased quantum mean estimator that reduces the variance quadratically faster than the classical empirical mean. No such estimator was known to exist prior to our work. These properties, which are of general interest, lead to better convergence guarantees within the paradigm of simulated annealing for computing partition functions.

%% file: Sources/introduction.tex
The Boltzmann–Gibbs distribution is a paradigmatic tool for modeling systems that obey the principle of maximum entropy.
It arises in several fields of research such as statistical mechanics~\cite{Geo11b,FV17b,Sin82b}, economic modeling~\cite{DY00j}, image processing~\cite{GG84j,Bes86j}, statistical learning theory~\cite{Cat04b}, etc.\
The probability assigned by the Gibbs distribution to each possible configuration of a system is inversely proportional to the exponential of its energy multiplied by the inverse temperature.
Mathematically, for a classical Hamiltonian $H : \Omega \ra \itv{0}{n}$ of degree $n$ specifying the energy level of each configuration $x$, the Gibbs distribution at inverse temperature $\beta$ is given by $\pi_{\beta}(x) = \frac{1}{Z(\beta)} e^{-\beta H(x)}$ where the normalization factor
  \begin{equation}
    Z(\beta) = \sum_{x \in \Omega} e^{-\beta H(x)}
  \end{equation}
is called the \emph{partition function}.
While it is often straightforward to evaluate the partition function at high temperature (when $\beta = 0$, it is just the number of possible configurations), the low-temperature regime captures ground state properties that are challenging to compute. For instance, $Z(\infty)$ can represent the cardinalities of exponentially large combinatorial structures (such as the number of colorings of a graph~\cite{Jer95j,SVV09j}, the volume of convex bodies~\cite{DFK91j,DF91c} or the permanent of non-negative matrices~\cite{JSV04j}) which are generally \#P-hard to compute exactly~\cite{Val79j,DF88j,JS93j}.

The standard approach for evaluating partition functions at low temperature is to resort to Markov chain Monte Carlo methods~\cite{JS96b}.
A celebrated line of works~\cite{VC72j,JVV86j,DF91c,BSVV08j,SVV09j} has shown how to turn the ability to efficiently \emph{sample} from the Gibbs distribution into that to efficiently \emph{approximate} the partition function.
At a high level, these works rely on the same two-stage simulated annealing algorithm.
First, they compute a short \emph{cooling schedule}, which is an increasing sequence of inverse temperature $0 = \beta_0 < \dots < \beta_{\ell} = \infty$ with limited fluctuations in Gibbs distributions between two consecutive values.
Next, the partition function at low temperature is expressed as a telescoping product
  \begin{equation}
    \label{Eq:telescop}
    Z(\infty) = Z(0) \prod_{i = 0}^{\ell-1} \frac{Z(\beta_{i+1})}{Z(\beta_i)}
  \end{equation}
which is approximated by using a suitable \emph{product estimator}. As an example, the seminal algorithm of {\v{S}}tefankovi{\v{c}}, Vempala and Vigoda~\cite{SVV09j} generates a so-called ``Chebyshev cooling schedule'' of length\footnote{We use the notation $\wbo{.}$ to hide polylogarithmic factors in the argument.} $\ell = \wbo{\sqrt{\log\abs{\Omega}}}$ (ignoring logarithmic dependences on the degree $n$, which is often on the order of $n \sim \log \abs{\Omega}$) where each ratio $Z(\beta_{i+1})/Z(\beta_i)$ is expressed as the expectation value of a random variable $X_i$ with bounded relative second moment $\ex{X_i^2} \leq \bo{\ex{X_i}^2}$. Such schedules are known to admit a product estimator that requires $\bo{\ell^2/\eps^2}$ classical Gibbs samples to estimate $Z(\infty) = Z(0) \cdot \ex{X_1} \cdots \ex{X_{\ell}}$ with relative error $\eps$.
Thus, if we let~$\delta$ denote the spectral gap of a (ergodic reversible) Markov chain generating samples from the considered Gibbs distributions, the overall cost of the algorithm presented in~\cite{SVV09j} is~$\wbo{\log\abs{\Omega}/(\eps^2\delta)}$.

The theory of quantum algorithms provides several directions for accelerating the computation of partition functions.
Quantum Markov chains~\cite{Sze04c,MNRS11j} can prepare coherent ``qsample'' encodings $\ket{\pi_{\beta}} = \sum_x \sqrt{\pi_{\beta}(x)} \ket{x}$ of the Gibbs distribution with a quadratic improvement in spectral gap for the rate of convergence (but an increase dependence on other parameters).
Quantum phase estimation~\cite{Kit95p} and amplitude estimation~\cite{BHMT02j} lead to quadratically better convergence rates for estimating expectation values~\cite{Ter99d,AW99p,Hei02j,BDGT11p,Mon15j,HM19c,Ham21c}.
Yet, while this may hint at the existence of an $\wbo{\sqrt{\log\abs{\Omega}}/(\eps\sqrt{\delta})}$ quantum algorithm for estimating partition functions, the best known algorithms~\cite{HW20c,AHN22j} still require a linear scaling $\wbo{\log\abs{\Omega}/(\eps\sqrt{\delta})}$ with the logarithm of the size of the state space.
This bottleneck is due to additional challenges posed by current quantum algorithmic techniques.
It is for instance significantly harder to prepare the qsample $\ket{\pi_{\beta}}$ (at low temperature) than to implement the reflection $\id - 2\proj{\pi_{\beta}}$ through it. This obstacle requires using \emph{nondestructive} procedures~\cite{MW05j,WA08j,WCNA09j,TOV+11j,ORR13j,HW20c} to recycle the same qsamples all along the algorithm, and to rely mostly on the reflection operator.
Another fundamental limitation faced by current best quantum mean estimators~\cite{Mon15j,HM19c,Ham21c} is the presence of \emph{biases} in the estimates that degrade the convergence guarantee of the product estimators.


\subsection{Contributions}

Our main contribution is to develop the first quantum algorithm for approximating Gibbs partition functions with a complexity scaling \emph{sublinearly} with respect to the logarithm of the size of the state space. More precisely, we prove the next theorem in Section~\ref{Sec:partition}.

\begin{rtheorem}[Theorem~\ref{Thm:QSA} {\normalfont (Informal)}]
  There is a quantum algorithm such that, given a Gibbs distribution generated by a Markov chain with spectral gap $\delta$, it computes an estimate $\td{Z}$ of the partition function at zero temperature satisfying $\abs{\td{Z} - Z(\infty)} \leq \eps Z(\infty)$ by using $\wbo[\big]{\log^{3/4}(|\Omega|)\log^{3/2}(n)/(\eps\sqrt{\delta})}$ steps of the quantum walk operator.
\end{rtheorem}

Our result reduces the polynomial dependence on $\log \abs{\Omega}$ by a factor of $1/4$ and it achieves state-of-the-art dependence on the spectral gap~$\delta$ and the accuracy~$\eps$ up to logarithmic factors. We provide a comparison with prior work in Table~\ref{Tab:partition}.

\begin{table}[htbp]
   \aboverulesep=0ex
   \belowrulesep=0ex
   \renewcommand{\arraystretch}{1.25}
  \centering\footnotesize
  \begin{tabular}{l|cc|c}
    \toprule
                     & Schedule generation               & Mean-value estimation & {\bf Total cost} \\ \midrule
      {\scriptsize\cite{DF91c,BSVV08j}} & $0$ (non-adaptive) & $\wbo[\big]{\log^2 \abs{\Omega}/(\eps^2 \delta)}$ & $\wbo[\big]{\log^2 \abs{\Omega}/(\eps^2 \delta)}$ \\
      {\scriptsize\cite{SVV09j,Hub15j,Kol18c}}  & $\wbo[\big]{\log \abs{\Omega}/\delta}$ & $\wbo[\big]{\log \abs{\Omega}/(\eps^2 \delta)}$ & $\wbo[\big]{\log \abs{\Omega}/(\eps^2 \delta)}$ \\ \midrule
      {\scriptsize\cite{WCNA09j}} & Use {\scriptsize\cite{BSVV08j}} & $\wbo[\big]{\log^2 \abs{\Omega}/(\eps\sqrt{\delta})}$ & $\wbo[\big]{\log^2 \abs{\Omega}/(\eps\sqrt{\delta})}$ \\
      {\scriptsize\cite{Mon15j}}  & Use {\scriptsize\cite{SVV09j}} & $\wbo[\big]{\log \abs{\Omega}/(\eps\sqrt{\delta})}$ & $\wbo[\big]{\log \abs{\Omega}(1/\delta + 1/\eps\sqrt{\delta})}$ \\
      {\scriptsize\cite{HW20c,AHN22j}} & $\wbo[\big]{\sqrt{\log \abs{\Omega}/\delta}}$ & Use {\scriptsize\cite{Mon15j}} & $\wbo[\big]{\log \abs{\Omega}/(\eps\sqrt{\delta})}$ \\ \midrule
      {\bf Our work} & Use {\scriptsize\cite{HW20c,AHN22j}} & $\wbo[\big]{\log^{3/4} \abs{\Omega}/(\eps\sqrt{\delta})}$ & $\wbo[\big]{\log^{3/4} \abs{\Omega}/(\eps\sqrt{\delta})}$ \\
    \bottomrule
  \end{tabular}
  \caption{\small Comparison of the complexity (in terms of Markov chain steps) needed to compute partition functions over a state space $\Omega$, where $\delta$ is the spectral gap of the Markov chain and $\eps$ is the accuracy parameter. Here, we omit the polylogarithmic dependencies on the degree $n$ of the partition function since in most applications $n \sim \log \abs{\Omega}, \poly(1/\delta)$. The first two rows are the classical algorithms.}
\label{Tab:partition}
\end{table}

The main technical ingredients of our work are new variants of the quantum phase and amplitude estimation algorithms (Theorems~\ref{Thm:unbiasPhase} and~\ref{Thm:nduAE}) that return (nearly) unbiased estimates and restore the qsamples used in the process with high probability. We extend these properties into new quantum mean estimators described in Section~\ref{Sec:mean}. In particular, the next result is crucial in the development of our improved product estimator (Theorem~\ref{Thm:prodEst}), and may be of independent interest. It shows that the relative error $\eps$ in bias can be decreased \emph{exponentially} quickly. In comparison, the estimate $\mut$ computed in previous works~\cite{Mon15j,HM19c,Ham21c} satisfies $\abs*{\mut - \mu} \leq \bo{\sigma/t}$ with high probability, which incurs a worst-case bias of roughly $\bo{\sigma/t}$.

\begin{rproposition}[Proposition~\ref{Prop:unEstim} {\normalfont (Informal)}]
  For any parameters $t,\eps,\vart$ there is a quantum algorithm such that, given a random variable~$X$ with variance at most $\td{\sigma}$, it computes a mean estimate $\td{\mu}$ with bias $\abs*{\ex*{\mut} - \mu} \leq \eps \vart$ and variance $\var{\mut} \leq \pt[\big]{\frac{\vart}{t}}^2$ by using one copy of a qsample $\ket{\pi_X}$ (restored at the end of the computation) and $\wbo{t\log(1/\eps)^2}$ applications of the reflection $\id - 2\proj{\pi_X}$.
\end{rproposition}

Finally, we provide applications of our work in Section~\ref{Sec:applications}, where we improved upon the best known algorithms for several approximate counting problems that can be phrased as partition functions. These results are described succinctly in Table~\ref{Tab:applications}.

\begin{table}[htbp]
   \aboverulesep=0ex
   \belowrulesep=0ex
   \renewcommand{\arraystretch}{1.3}
  \centering\footnotesize
  \begin{tabular}{l@{\hskip 0.2cm}ccc}
    \toprule
                     & Colorings, Ising model, & \multirow{1.7}{*}{Matchings} & Volume of \\ [-2mm]
                     & Independent sets & & convex body \\ \midrule
      {\scriptsize \cite{SVV09j,CLV21c}} & $\wbo{\abs{V}^2/\eps^2}$ & $\wbo{\abs{V} \abs{E}/\eps^2}$ & \textendash \\ [-1mm]
      {\scriptsize \cite{CV18j,JLLV21c}} & \textendash & \textendash & $\wbo{d^{3.5} + d^3/\eps^2}$ \\ \midrule
      {\scriptsize \cite{HW20c,AHN22j}} & $\wbo{\abs{V}^{1.5}/\eps}$ & $\wbo{\abs{V}\abs{E}^{0.5}/\eps}$ & \textendash \\ [-1mm]
      \cite{CCH+19p} & \textendash & \textendash & $\wbo{d^3 + d^{2.5}/\eps}$ \\ \midrule
      {\bf Our work} & $\wbo{\abs{V}^{1.25}/\eps}$ & $\wbo{\abs{V}^{0.75}\abs{E}^{0.5}/\eps}$ & $\wbo{d^3 + d^{2.25}/\eps}$ \\
    \bottomrule
  \end{tabular}
\caption{\small Comparison of the best known classical (rows 1--2) and quantum (rows 3--5) algorithms for approximate counting problems. Here, $\eps$ is the accuracy parameter, $V$ is the vertex set, $E$ is the edge set and $d$ is the dimension of the convex body. The complexities are expressed in terms of number of random or quantum walk steps (for volume estimation, each step requires one query to a membership oracle).}
\label{Tab:applications}
\end{table}


\subsection{Proof overview}

We give a high-level description of the algorithms developed in this paper. We first explain the main differences with previous works~\cite{WCNA09j,Mon15j,HW20c,AHN22j} that allow us to obtain a sublinear algorithm for estimating partition functions. Next, we present new variants of \emph{phase estimation}, \emph{amplitude estimation} and \emph{quantum mean estimation} that are needed for implementing our approach.

\paragraph*{Estimating partition functions (Section~\ref{Sec:partition}).}
The current bottleneck in quantizations of the classical simulated annealing algorithm of {\v{S}}tefankovi{\v{c}}, Vempala and Vigoda~\cite{SVV09j} lies in estimating the expectation of a Chebyshev cooling schedule for the telescoping product displayed in Equation~(\ref{Eq:telescop}). Abstractly, this problem amounts to estimating with relative error $\eps$ the expectation of the product $X = X_1 \cdots X_{\ell}$ of~$\ell$ independent random variables with bounded relative variance $\var{X_i} \leq \bo{\ex{X_i}^2}$. By Bernoulli's inequality, one can show that estimating the expectation of each inner random variable~$X_i$ with relative error~$\bo{\eps/\ell}$ and taking the product of the estimates is sufficient for that purpose. Montanaro~\cite{Mon15j} developed a quantum algorithm for computing each such inner estimate in time~$\wbo{\ell/\eps}$, quadratically faster than it is possible classically, leading to an overall complexity of $\wbo{\ell^2/\eps}$. Nevertheless, a different \emph{classical} estimator from Dyer and Frieze~\cite{DF91c} provides the same quadratic complexity with respect to the schedule length: replace the $\eps/\ell$-error inner estimates with \emph{unbiased} estimates of variance $\bo{(\eps^2/\ell) \cdot \ex{X_i}^2}$. The latter estimates can simply be obtained by averaging $\bo{\ell/\eps^2}$ samples due to the bounded variance property. The overall variance of this new product estimator is shown to be $\bo{\pt{\eps \ex{X}}^2}$ which, by Chebyshev's inequality and the fact that it is unbiased, implies that the overall relative error is~$\eps$. We follow a similar approach in our work to speed up the estimation of~$\ex{X}$. Our main contribution is a new quantum estimator that reduces the variance quadratically faster than classically, while keeping the estimates (nearly) unbiased. This allows us to compute each inner estimate in time~$\wbo{\sqrt{\ell}/\eps}$, leading to an overall complexity of~$\wbo{\ell^{3/2}/\eps}$. Our estimator requires new unbiased variants of the \emph{phase} and \emph{amplitude estimation} algorithms, which are sketched in the next paragraphs. Additionally, these algorithms are made \emph{nondestructive} (i.e. they restore a copy of their starting state) to ensure the reusability of the \emph{qsamples} which are an expensive resource in our applications.

\paragraph*{Unbiased phase estimation (Section~\ref{Sec:unbiasedPE}).}
The \emph{phase estimation} problem asks the question of estimating the eigenphase $\theta \in \bp{0,1}$ of a quantum state $\ket{\psi}$ satisfying $U \ket{\psi} = e^{2\pi i \theta} \ket{\psi}$ for a given unitary operator~$U$. Kitaev's algorithm~\cite{Kit95p} returns an estimate $\tht$ within distance~$1/t$ from~$\theta$ with constant success probability by using $\bo{t}$ controlled-$U$ operations. By standard success amplification techniques, the variance of this estimate can be guaranteed to be of the order of $1/t^2 + \eps$ at the cost of an $\bo{\log 1/\eps}$ overhead in complexity. One simple idea to make this estimator unbiased is to consider the \emph{phase-shifted} unitary $U_{\varphi} \ket{\psi} = e^{2\pi i (\theta+\varphi)} \ket{\psi}$ for a random~$\varphi$, run phase estimation on it and subtract the shift $\varphi$ from the estimate. Although it provides an unbiased estimate of~$e^{2\pi i \theta}$ (and, therefore, of $\sin^2(\theta)$), we do not know how to reduce its variance without introducing a new bias. We overcome this difficulty by exploiting more properties of the output distribution of phase estimation and looking for an unbiased estimate of~$\theta$ directly. First, we show that it is possible to compute with high probability the \emph{exact} value of the first $\log t$ bits of $\theta + \varphi = 0.b_1b_2\dots$ for a well-chosen phase shift $\varphi$ (steps~1--5 in Algorithm~\ref{Alg:UPE}). This relies on the ability to pinpoint exactly the interval $\bp*{\frac{k}{t},\frac{k+1}{t}}$ for $k \in \itv{0}{t-1}$ containing the phase $\theta + \varphi$, as long as the latter is not too close to an interval endpoint (which is avoided by the choice of~$\varphi$). Secondly, we look for a low-variance unbiased estimate of the remaining bits $0.0^{\log t} b_{\log t +1}b_{\log t +2}\dots$ (steps~6--7 in Algorithm~\ref{Alg:UPE}). For that, we consider the exponentiated unitary $U_{\varphi}^{t}$ which amplifies the eigenphase to $\lambda = 0.b_{\log t +1}b_{\log t +2}\dots$. Using \emph{phase-to-amplitude} conversion techniques, we transform it into a new unitary~$V$ that encodes~$\lambda$ into the amplitude of a new auxiliary quantum register. This register can simply be measured to obtain an unbiased binary estimate $\td{\lambda} \in \{-1,1\}$ of~$\lambda$. The variance of the correction term $\td{\lambda}/t$ added to the previously computed bits $0.b_1b_2\dots b_{\log t}$ is thus of order~$\bo{1/t^2}$. We also have to be careful that~$\lambda$ is bounded away from $0$ and $1$ in order for the conversion techniques to succeed. This is again guaranteed by carefully choosing the phase shift. The final algorithm is summarized in Theorem~\ref{Thm:unbiasPhase}.

\paragraph*{Nondestructive amplitude estimation (Appendix~\ref{App:nondestrAE}).}
The \emph{amplitude estimation} algorithm~\cite{BHMT02j} approximates the measurement probability $p = \norm{\Pi \ket{\psi}}^2$ for a state $\ket{\psi}$ and a projector $\Pi$, given access to the reflection operators $R_{\psi} = \id - 2\proj{\psi}$ and $R_{\Pi} = \id - 2\Pi$. This algorithm relies on the observation that the Grover operator $G = -R_{\psi}R_{\Pi}$ has two conjugate eigenvalues $e^{\pm 2i\arcsin(\sqrt{p})}$ with $\ket{\psi}$ being a superposition over the two corresponding eigenvectors. In its original form~\cite{BHMT02j}, the algorithm consists of running phase estimation on the latter state and applying the function $x \mapsto \sin^2(\pi x)$ on the measured register. Although the post-measurement state is still a superposition over the two above eigenvectors, this process can introduce a relative phase that makes it far from the initial state~$\ket{\psi}$. Harrow and Wei~\cite{HW20c} argued that amplitude estimation can be made \emph{nondestructive} (i.e. it restores the state~$\ket{\psi}$) by using an \emph{uncomputation} trick inspired by the Marriott-Watrous scheme~\cite{MW05j}. We discuss this result in Appendix~\ref{App:nondestrAE} and we provide a new proof that avoids some complications in the analysis.

\paragraph*{Nondestructive unbiased amplitude estimation (Section~\ref{Sec:unbiasedAE}).}
We aim at combining the above two properties in order to obtain a variant of amplitude estimation that is both nondestructive and unbiased. Unfortunately, the properties of the Grover operator do not allow us to simply use the unbiased phase estimation on it. Indeed, it would require the ability to shift its two eigenphases by \emph{opposite} values simultaneously. Moreover, it would return an unbiased estimate of $\arcsin(\sqrt{p})$ instead of~$p$. We overcome these two obstacles by first applying the function $x \mapsto \sin^2(\pi x)$ on the eigenphases of the Grover operator (via amplitude-to-phase conversion techniques~\cite{GAW17p}), which produces a new unitary~$G'$ with ``merged'' eigenvalue~$e^{2\pi ip}$ and eigenvector~$\ket{\psi}$. We next apply the unbiased phase estimation on~$G'$ and~$\ket{\psi}$, which is inherently nondestructive since the latter state is now an eigenvector. Yet, this unbiased estimator poses a new problem: its variance is worse than that of the original amplitude estimation algorithm (roughly, $1/t^2$ vs.\ $p/t^2 + 1/t^4$), which is due to encoding $p$ instead of $\arcsin(\sqrt{p})$ into the phase. Our solution to this problem is to first compute a rough \emph{biased} amplitude estimate $q = \ta{\max\pt{p,1/t^2}}$ (via standard nondestructive amplitude estimation) and to apply the new function $x \mapsto \sin^2(\pi x)/q$ on the eigenphases of the Grover operator (via linear amplitude amplification~\cite{Low17d}). Since the obtained operator will require time $\bo{1/\sqrt{q}}$ to be simulated, we are left with $\bo{\sqrt{q}t}$ steps available to perform unbiased phase estimation on it. The variance is then shown to be of the order of~$p/t^2$ when $p = \bo{1/t^2}$ and~$q/t^2$ otherwise. A side tool of independent interest (needed when $p = \bo{1/t^2}$) is a \emph{nondestructive coin flipping} algorithm (Proposition~\ref{Prop:coin}) for sampling from the Bernoulli variable of parameter~$\norm{\Pi \ket{\psi}}^2$ in constant expected time.

\paragraph*{Nondestructive unbiased mean estimator (Section~\ref{Sec:mean}).}
Our final step is to generalize the amplitude estimation algorithm to estimate the expectation of an arbitrary random variable $X$ whose distribution is encoded into a qsample $\ket{\pi_X} = \sum_x \sqrt{\pr{X = x}} \ket{x}$ (amplitude estimation handles the case of Bernoulli distributions). The estimator shall meet the three requirements of being unbiased, nondestructive and low-variance in order to be applied efficiently to the partition function problems. We build upon the quantum estimators from~\cite{Hei02j,Mon15j,HM19c,Ham21c} that already satisfy the low-variance property. In the latter works, the estimation of~$\ex{X}$ is reduced to estimating the means of a series of Bernoulli distributions whose parameters are determined by certain \emph{quantiles} of the input distribution. We follow the same approach, by instead using the nondestructive unbiased amplitude estimation algorithm in the reduction. This is not enough yet to make the estimator nondestructive and unbiased. First, we modify a recentering step in~\cite{Mon15j,Ham21c} that computed a rough mean estimate by averaging classical samples (obtained from destructive measurements of copies of $\ket{\pi_X}$). We develop a new \emph{median estimator} for this purpose that needs not destroy any copy of $\ket{\pi_X}$ (Proposition~\ref{Prop:medi}). Secondly, the existing estimators require \emph{truncating} the random variable~$X$ by replacing its outcomes larger than a certain threshold value with zero. We suppress most of the truncation-induced bias by adding extra Bernoulli distributions to the reduction. The sum of these added estimates has low variance and is equal (in expectation) to most of the bias introduced by previous approaches. The final estimator is summarized in Proposition~\ref{Prop:unEstim}.

%% file: Sources/unbiasedPE.tex
In this section, we describe a variant of phase estimation that returns a (nearly) unbiased estimate. We use the next notation for separating the most and least significant bits of a phase $\theta \in \bc{0,1}$.

\begin{definition}
  Let $\theta \in \bc{0,1}$ be a real number with binary expansion $\vp = 0.\vp_1 \vp_2 \vp_3 \dots$ where $\vp_j \in \rn$. Then, for any integer $\tau \geq 0$, we denote $\vp_{\leq \tau} = 0.\vp_1 \vp_2 \dots \vp_{\tau}$ and $\vp_{> \tau} = 0.\vp_{\tau+1}\vp_{\tau+2}\dots$ such that $\vp = \vp_{\leq \tau} + 2^{-\tau} \vp_{> \tau}$.
\end{definition}

Suppose that we can call a controlled unitary $U$, and that we have a copy of an eigenstate~$\ket{\psi}$, satisfying $U\ket{\psi} = e^{2\pi i\theta}\ket{\psi}$. Then, we can implement the \emph{phase-shift} operation $e^{2\pi i \varphi} U \ket{\psi} = e^{2\pi i(\theta + \varphi)}\ket{\psi}$ for any $\varphi$ with one call to $U$, and the \emph{bit-shift} operation $U^{2^{\tau}} \ket{\psi} = e^{2\pi i \theta_{> \tau}}\ket{\psi}$ with~$2^{\tau}$ calls to $U$. We make use of these two operations to design the following unbiased, nondestructive, low-variance phase estimation algorithm.

\begin{algorithm}[H]
  \caption{Unbiased low-variance phase estimation, $\upe(U,\ket{\psi},t,\eps)$.}
  \label{Alg:UPE}
  \begin{algorithmic}[1]
    \For{$k = 0,\dots,8$}
      \State Set $\varphi = k/(16t')$ and $U_{\varphi} = e^{2\pi i \varphi} U$ where $t' = 8t$.
      \State\StateInd{Run \hyperref[Lem:PE]{phase estimation} $N = \ceil{200 \log(72t' /\eps)}$ times on the phase-shifted unitary $U_{\varphi}$ with initial state $\ket{\psi}$ and time $t'$ to collect $N$ independent estimates $\td{\varphi}_1,\dots,\td{\varphi}_{N}$ of $\varphi$. Compute the frequencies~$(\td{p}(i))_{i \in \itv{0}{t'-1}}$ where $\td{p}(i) = \abs{\set{j \in \set{1,\dots,N} : \td{\varphi}_j = i/t'}}$.\vspace{2mm}}
      \If{there is an index $i$ such that $\td{p}(i),\td{p}(i+1) \geq 0.17$ and $(i/t')_{> \tau} \in (4/7,5/7)$}
      \State Stop the for loop and set $\tht_{\varphi} = (i/t')_{\leq \tau}$.
      \EndIf
    \EndFor
    \State Apply \hyperref[Lem:PrO]{phase-to-amplitude conversion} on the exponentiated unitary $U_{\varphi}^t$ (for the last value of~$\varphi$ computed in step 2) with precision $\eps/4$ to obtain a unitary $V$ such that
      \[V \pt[\big]{\ket{\psi}\ket{0}^{\otimes a}} = \ket{\psi} \pt[\big]{\sqrt{1-p'}\ket{0}^{\otimes a} + \sqrt{p'}\ket{\Phi^{\perp}}}.\]
    \State Compute the state $V \pt[\big]{\ket{\psi}\ket{0}^{\otimes a}}$ and measure its last $a$ qubits in the standard basis. If the result is the all-zero string then set $b = 0$ else set $b = 1/(2\pi)$.
    \State Output $\tht = \tht_{\varphi} + b/t - \varphi$.
  \end{algorithmic}
\end{algorithm}

\begin{theorem}[\sc Unbiased low-variance phase estimation]
  \label{Thm:unbiasPhase}
  Let $U$ be a unitary operator with an eigenvector $\ket{\psi}$ such that $U \ket{\psi} = e^{2\pi i \theta} \ket{\psi}$ where $\theta \in \bc{0,1/2}$. Given $t = 2^{\tau} \geq 1$ and $\eps \in (0,1)$, the \emph{unbiased low-variance phase estimation} algorithm $\upe(U,\ket{\psi},t,\eps)$ (Algorithm~\ref{Alg:UPE}) outputs an estimate $\tht \in \bc{-1,1}$ such that
    \[\abs*{\ex{\tht} - \theta} \leq \eps \quad \text{and} \quad \var{\tht} \leq \frac{1}{t^2} + \eps.\]
  The algorithm needs one copy of $\ket{\psi}$, which is restored at the end of the computation with probability at least $1-\eps$, and $\bo{t \log(t/\eps)}$ controlled-$U$ operations.
\end{theorem}

\begin{proof}
  We decompose the analysis of Algorithm~\ref{Alg:UPE} into three parts. First, we show that the two phases $\varphi \in (0,1/t)$ and $\tht_{\varphi} \in (0,1)$ obtained at the end of steps 1--5 satisfy the equations
    \begin{equation}
      \label{Eq:EPE}
      \tht_{\varphi} = (\theta + \varphi)_{\leq \tau} \qquad \text{and} \qquad (\theta + \varphi)_{> \tau} \in (1/2,3/4)
     \end{equation}
  with probability at least $1-\eps/8$. We refer to this part as \emph{exact phase estimation} since $\tht_{\varphi}$ is exactly equal to the first $\tau$ bits of the shifted phase $\theta + \varphi$. Next, conditioned on any pair $(\varphi,\tht_{\varphi})$ satisfying the above equation, we prove that the estimate $b$ obtained at the end of steps 6--7 satisfies
    \begin{equation}
      \label{Eq:UPE}
      \abs{\ex{b \given (\varphi,\tht_{\varphi})} - (\theta + \varphi)_{> \tau}} \leq \eps/2 \qquad \text{and} \qquad \var{b \given (\varphi,\tht_{\varphi})} \leq 1.
     \end{equation}
  We refer to this part as \emph{unbiased phase estimation}. Finally, we analyze the expectation and variance of the final estimate constructed at step 8.

  \textit{Exact phase estimation (steps 1--5).} For each phase shift $\varphi$ and integer $i \in \itv{0}{t'-1}$, let $\eps^{\varphi}_i$ denote the distance between the angles $\theta + \varphi$ and $i/t'$ along the unit circle. The two smallest distances are achieved by $\eps^{\varphi}_{i^{\star}},\eps^{\varphi}_{i^{\star}+1} \leq 1/t'$ where $i^{\star} = t' (\theta + \varphi)_{\leq 3\tau}$, whereas the other indices $j \notin \set{i^{\star},i^{\star}+1}$ are at distance at least $\eps^{\varphi}_j > 1/t'$. By Hoeffding's inequality and a union bound over the $t'$ possible indices, at each round of step~3 the empirical frequency vector $\td{p}$ is at distance $\norm{p - \td{p}}_{\infty} \leq 0.05$ from the probability vector~$p$ with probability at least $1-\eps/72$. Since step~3 is executed at most 9 times, all of the frequency vectors computed by the algorithm satisfy this bound with probability at least $1-\eps/8$. On the other hand, there is at least one value of~$\varphi$ (among the 9 possible ones) such that $\eps^{\varphi}_{i^{\star}},\eps^{\varphi}_{i^{\star}+1} \leq 5/(8t')$ and $(\theta + \varphi)_{> \tau} \in (1/2,3/4)$. The algorithm will correctly detect such a value because of the gaps proved in Corollary~\ref{Cor:phaseProba} between the probabilities $p(i^{\star}),p(i^{\star}+1)$ and the probabilities $p(j)$ for $j \notin \set{i^{\star},i^{\star}+1}$. This proves Equation~(\ref{Eq:EPE}).

  \textit{Unbiased phase estimation (steps 6--7).} The state $\ket{\psi}$ is an eigenvector of $U^t_{\varphi}$ such that $U^t_{\varphi} \ket{\psi} = e^{2\pi i (\theta + \varphi)_{> \tau}} \ket{\psi}$. If the pair $(\varphi,\tht_{\varphi})$ obtained at the end of steps 1--6 satisfies Equation~(\ref{Eq:EPE}) then, by Lemma~\ref{Lem:PrO}, the amplitude~$\sqrt{p'}$ obtained at step~6 is at distance at most $\abs{\sqrt{p'} - \sqrt{2\pi (\theta + \varphi)_{> \tau}}} \leq \eps/4$ from the shifted phase. Thus, conditioned on any fixed correct pair $(\varphi,\tht_{\varphi})$, we get $\abs{\ex{b \given (\varphi,\tht_{\varphi})} - (\theta + \varphi)_{> \tau}} = \abs{p' - (\theta + \varphi)_{> \tau}} \leq \eps/2$ and $\var{b \given (\varphi,\tht_{\varphi})} \leq 1$ for the estimate~$b$ computed at step 7. This proves Equation~(\ref{Eq:UPE}).

  \textit{Low-variance (step 8).} The expectation and variance of the final estimate $\tht$ are $\abs{\ex{\tht \given (\varphi,\tht_{\varphi})} - \theta} \leq \eps/2$ and $\var{\tht \given (\varphi,\tht_{\varphi})} \leq 1/t^2$ conditioned on any pair $(\varphi,\tht_{\varphi})$ satisfying Equation~(\ref{Eq:EPE}). Since $\abs{\theta} \leq 1$ and Equation~(\ref{Eq:EPE}) holds with probability at least $1-\eps/8$, we obtain by the law of total expectation that $\abs{\ex{\tht} - \theta} \leq \eps/2 + \max\abs{\tht - \theta} \cdot \eps/8 \leq 3\eps/4$ and by the law of total variance that $\var{\tht} = \ex{\var{\tht \given (\varphi,\tht_{\varphi})}} + \var{\ex{\tht \given (\varphi,\tht_{\varphi})}} \leq 1/t^2 + \eps/8 + \ex{\abs{\ex{\tht \given (\varphi,\tht_{\varphi})} - \theta}^2} \leq 1/t^2 + \eps/8 + (\eps/2)^2 + 4\cdot\eps/8 \leq 1/t^2 + 7\eps/8$. The number of controlled-$U$ operations used by the algorithm is $\bo{t \log(t/\eps)}$ at steps 1--5 and $\bo{t \log(1/\eps)}$ at steps 6--7. Moreover, the state $\ket{\psi}$ is preserved by the phase estimation and the phase-to-amplitude conversion algorithms.
\end{proof}

%% file: Sources/unbiasedAE.tex
In this section, we describe a variant of amplitude estimation that returns a (nearly) unbiased estimate and restores the initial state with high probability. We start with a simple procedure, adapted from the Marriott-Watrous scheme~\cite{MW05j}, to sample nondestructively from a Bernoulli variable of parameter $\norm{\Pi \ket{\psi}}^2$ given a projector $\Pi$ and a single copy of a state $\ket{\psi}$.

\begin{algorithm}[H]
	\caption{Nondestructive coin flip}
  \label{Alg:coin}
	\begin{algorithmic}[1]
		\State Measure the state $\ket{\psi}$ according to the projectors $\set{\Pi, \id-\Pi}$. Set $b = 1$ if the result is $\Pi \ket{\psi}$ and $b = 0$ if it is $(\id-\Pi) \ket{\psi}$.
    \State Measure the residual state according to the projectors $\set{\proj{\psi}, \id-\proj{\psi}}$.
		\While{the residual state is not $\ket{\psi}$}
      \State Measure the residual state according to the projectors $\set{\Pi, \id-\Pi}$.
		  \State Measure the residual state according to the projectors $\set{\proj{\psi}, \id-\proj{\psi}}$.
		\EndWhile
		\State Output $b$ and $\ket{\psi}$.
	\end{algorithmic}
\end{algorithm}

\begin{proposition}[\sc Nondestructive coin flip]
  \label{Prop:coin}
  Let $\ket{\psi}$ be a quantum state and~$\Pi$ be a projection operator. The \emph{nondestructive coin flip} algorithm (Algorithm~\ref{Alg:coin}) outputs a Boolean value $b \in \rn$ such that $\ex{b} = \norm{\Pi \ket{\psi}}^2$.
  The algorithm needs one copy of $\ket{\psi}$, which is restored at the end of the computation, and $\bo{1}$ applications in expectation of the reflection operators $\id - 2\proj{\psi}$ and $\id - 2\Pi$.
\end{proposition}

\begin{proof}
  Let $p = \norm{\Pi \ket{\psi}}^2$ denote the probability to be estimated and define $r = 2p(1-p)$. By the standard Marriott-Watrous analysis~\cite{MW05j}, the residual state at the end of step 2 is $\ket{\psi}$ with probability $1-r$ (in which case steps 3--5 need not be executed) and the probability that one iteration of steps 4--5 succeeds in restoring $\ket{\psi}$ is $r$. Thus, the expected number of iterations needed to restore $\ket{\psi}$ is $r^2 \sum_{i = 0}^{\infty} (i+1) (1-r)^i = 1$.
\end{proof}

We now extend the above algorithm to a low-variance unbiased estimator.

\begin{algorithm}[H]
	\caption{Nondestructive unbiased amplitude estimation, $\nduae(\ket{\psi},\Pi,t,\eps)$.}
  \label{Alg:nduAE}
	\begin{algorithmic}[1]
    \State Run \hyperref[Prop:ndAE]{nondestructive amplitude estimation} $\ndae(\ket{\psi},\Pi,2t,\eps/8)$ on $\ket{\psi}$ with projector~$\Pi$ to obtain a (biased) amplitude estimate $q$ of $p = \norm{\Pi \ket{\psi}}^2$.
		\State Let $\tau = \max\set*{1,\frac{1}{4}\min\set{t,1/\sqrt{q}}}$. Run \hyperref[Lem:AA-SVT]{linear amplitude amplification} on $\ket{\psi}$ with projector~$\Pi$, time $\tau$ and precision $\eps/40$ to obtain a unitary $V_{\tau,\eps/40}$ preparing a state $\ket{\psi'} = V_{\tau,\eps/40} \ket{\psi}$ such that
      \[\abs[\big]{\norm{\Pi \ket{\psi'}} - \tau \norm{\Pi \ket{\psi}}} \leq \eps/40.\]
		\If{$q \leq 1/t^2$}
      \State\StateInd{Measure the state $\ket{\psi'}$ according to the projectors $\set{\Pi, \id-\Pi}$ by using the \hyperref[Prop:coin]{nondestructive coin flip} algorithm. Let $b \in \rn$ denote the Boolean value it returns.}
      \State Output $\td{p} = b/\tau^2$.
		\Else
      \State\StateInd{Run \hyperref[Lem:PO]{amplitude-to-phase conversion} on the unitary $V_{\tau,\eps/40}$ with precision $\eps' = \om[\big]{\frac{\tau^2}{t^2\log(t/(\eps\tau))}}$ to obtain a unitary $\pora_{\eps'}$ such that}
        \[\norm{\pora_{\eps'} \pt[\big]{\ket{\psi'}\ket{0}^{\otimes a}} - e^{i \norm{\Pi \ket{\psi'}}^2}\ket{\psi'}\ket{0}^{\otimes a}} \leq \eps'.\]
      \State\StateInd{Run the \hyperref[Thm:unbiasPhase]{unbiased phase estimation} algorithm $\upe(\pora_{\eps'},\ket{\psi'}\ket{0}^{\otimes a},7t/\tau,\eps/90)$ on~$\pora_{\eps'}$ with (approximate) eigenstate $\ket{\psi'}\ket{0}^{\otimes a}$ to compute an eigenphase estimate $\td{p}\,'$.}
      \State Output $\td{p} = 2\pi \td{p}\,' / \tau^2$.
    \EndIf
	\end{algorithmic}
\end{algorithm}

\begin{theorem}[\sc Nondestructive unbiased amplitude estimation]
  \label{Thm:nduAE}
  Let $\ket{\psi}$ be a quantum state and~$\Pi$ be a projection operator with $p = \norm{\Pi \ket{\psi}}^2$. Given $t \geq 4$ and $\eps \in (0,1)$, the \emph{nondestructive unbiased amplitude estimation} algorithm $\nduae(\ket{\psi},\Pi,t,\eps)$ (Algorithm~\ref{Alg:nduAE}) outputs an estimate $\td{p} \in [-2\pi,2\pi]$ such that
    \[\abs*{\ex*{\td{p}} - p} \leq \eps \quad \text{and} \quad \var{\td{p}} \leq \frac{91p}{t^2} + \eps.\]
  The algorithm needs one copy of $\ket{\psi}$, which is restored at the end of the computation with probability at least $1-\eps$, and $\bo{t \log\log(t) \log(t/\eps)}$ applications in expectation of the reflection operators $\id - 2\proj{\psi}$ and $\id - 2\Pi$.
\end{theorem}

\begin{proof}
  By a simple application of Proposition~\ref{Prop:ndAE}, with probability at least $1-\eps/8$, the amplitude estimate $q$ computed at step~1 satisfies
    \[\td{q} = 0 \ \text{if $p \leq 1/(4t^2)$,} \qquad q \leq 7p \ \text{if $p \geq 1/(4t^2)$,} \qquad q \geq p/4 \ \text{if $q > 1/t^2$.}\]
  We analyze the rest of the algorithm conditioned on the above inequalities being true. Similarly to the proof of Theorem~\ref{Thm:unbiasPhase}, by the laws of total expectation and total variance, when removing the condition the expectation and variance of the estimate $\td{p}$ will differ respectively by at most $\eps/4$ and $7\eps/8$ from the values computed below.

  The amplification parameter $\tau$ at step~2 satisfies $\tau\norm{\Pi \ket{\psi}} = \tau \sqrt{p} \leq 1/2$ when $\tau > 1$, as required by Lemma~\ref{Lem:AA-SVT}. Thus, the norm of the post-amplification state $\ket{\psi'}$ projected along $\Pi$ is
    \begin{equation}
      \label{Eq:linamp}
      \abs[\big]{\norm{\Pi \ket{\psi'}}^2 - \tau^2 p} = \abs[\big]{\norm{\Pi \ket{\psi'}} + \tau \norm{\Pi \ket{\psi}}} \cdot \abs[\big]{\norm{\Pi \ket{\psi'}} - \tau \norm{\Pi \ket{\psi}}} \leq \eps/20.
    \end{equation}
  We split the analysis of the subsequent steps of Algorithm~\ref{Alg:nduAE} into two cases depending on whether the biased estimate $q$ is smaller or larger than $1/t^2$.

  \textit{If $q \leq 1/t^2$ (steps 3--5).} We have $\tau = t/4$ and, by Proposition~\ref{Prop:coin}, the Boolean value $b$ at step~4 is sampled from a Bernoulli distribution of parameter $\ex{b} = \norm{\Pi \ket{\psi'}}^2$. Thus, the expectation of the output estimate satisfies $\abs*{\ex*{\td{p}} - p} = \abs[\big]{\pt{\norm{\Pi \ket{\psi'}}/\tau}^2 - p} \leq \eps/(20\tau^2) \leq \eps/2$ and its variance is $\var{\td{p}} \leq \ex{\td{p}^2} = \norm{\Pi \ket{\psi'}}^2/\tau^4 \leq (\tau^2 p + \eps/20)/\tau^4 \leq p/\tau^2 + \eps/(20\tau^4) \leq 16p/t^2 + \eps/8$.

  \textit{If $q > 1/t^2$ (steps 6--9).} We have $\tau = \max\set*{1,1/(4\sqrt{q})}$. We start by analysing the estimate obtained when the unitary $\pora_{\eps'}$ used in phase estimation (step~8) is replaced with a unitary~$\pora$ that perfectly satisfies
    \[\pora \pt[\big]{\ket{\psi'}\ket{0}^{\otimes a}} = e^{i \norm{\Pi \ket{\psi'}}^2}\ket{\psi'}\ket{0}^{\otimes a}\]
  and is at distance at most $\norm{\pora - \pora_{\eps'}} \leq \eps'$ from $\pora_{\eps'}$ (such a unitary exists since $\norm{\pora_{\eps'} \pt[\big]{\ket{\psi'}\ket{0}^{\otimes a}} - e^{i \norm{\Pi \ket{\psi'}}^2}\ket{\psi'}\ket{0}^{\otimes a}} \leq \eps'$ by Lemma~\ref{Lem:PO}). We rename the phase and output estimates computed at steps 8 and 9 to $\ha{p}\,'$ and $\ha{p}$ respectively in this case. Observe that the eigenphase $\norm{\Pi \ket{\psi'}}^2/(2\pi)$ satisfies the boundedness requirement of Theorem~\ref{Thm:unbiasPhase} since it is at most $(\tau^2 p + \eps/20)/(2\pi) \leq 1/2$ by Equation~(\ref{Eq:linamp}) and $\tau = \max\set*{1,1/(4\sqrt{q})} \leq \max\set*{1,1/(2\sqrt{p})}$. Thus, the expectation of the estimate $\ha{p}\,'$ computed at step~8 satisfies $\abs*{\ex*{2\pi \ha{p}\,'} - \tau^2 p} \leq \abs[\big]{\ex*{2\pi \ha{p}\,'} - \norm{\Pi \ket{\psi'}}^2} + \abs*{\norm{\Pi \ket{\psi'}}^2 - \tau^2 p} \leq 2\pi \eps/90 + \eps/20$, where the first step is by the triangle inequality and the second step is by Theorem~\ref{Thm:unbiasPhase} and Equation~(\ref{Eq:linamp}). The variance is at most $\var{\ha{p}\,'} \leq \tau^2/(7 t)^2 + \eps/90$ by Theorem~\ref{Thm:unbiasPhase}. Since the final estimate is the scalar multiple $\ha{p} = 2\pi \ha{p}\,' / \tau^2$ of $\ha{p}\,'$, its expectation and variance are $\abs*{\ex*{\ha{p}} - p} \leq 2\pi(2\pi \eps/90 + \eps/20)/ \tau^2 \leq \eps/4$ and $\var{\ha{p}} \leq 4\pi^2/(7 t\tau)^2 + 4\pi^2\eps/(90\tau^4) \leq 91p/t^2 + \eps/9$.

  Consider now the case where $U$ is replaced with the actual unitary $\pora_{\eps}$. Let $T$ denote the number of times the operator $\pora_{\eps}$ is used by the unbiased phase estimation procedure at step~8. By Theorem~\ref{Thm:unbiasPhase}, we have $T = \bo{(t/\tau)\log(t/(\eps \tau))}$. By standard approximation arguments (see~\cite[p.194]{NC11b} for instance), the distribution of the actual estimate~$\td{p}$ is at distance at most $2 T \norm{\pora - \pora_{\eps'}}$ (in $\ell_{\infty}$-norm) from that of $\ha{p}$. Moreover, by a simple analysis of Algorithm~\ref{Alg:UPE}, both estimates are distributed in $[-1,1]$ and can take at most $36\pi t/\tau$ different values. We conclude that $\abs*{\ex*{\td{p}} - p} \leq \abs*{\ex*{\ha{p}} - p} + 2 T \norm{\pora - \pora_{\eps'}} \cdot 36\pi t/\tau \leq \eps/4 + 2T\eps' \cdot 36\pi t/\tau \leq \eps/2$ and $\var{\td{p}} \leq \var{\ha{p}} + 8 T \norm{\pora - \pora_{\eps'}} \cdot 36\pi t/\tau \leq 91 p/t^2 + \eps/9 + 8 T \eps' \cdot 36\pi t/\tau \leq 91 p/t^2 + \eps/8$ by choosing a value $\eps' = \ta[\big]{\frac{\tau^2}{t^2\log(t/(\eps\tau))}}$.
\end{proof}

%% file: Sources/mean.tex
We construct first a quantum estimator that decreases the variance quadratically faster than classically while keeping the estimate (nearly) unbiased. Moreover, it requires a single qsample, which is not destroyed with high probability. The main ingredient is the nondestructive amplitude estimation algorithm constructed in the previous section. The estimator is based on expressing the mean as a linear combination of carefully chosen Bernoulli random variables (with parameters $\mut^+_j$ and $\mut^-_j$ defined at steps 4 and 7 of Algorithm~\ref{Alg:unEstim}) which are related to different truncations of the input random variable. We use a more sophisticated truncation pattern than in previous work~\cite{Hei02j,Mon15j,Ham21c,CHJ22c} to ensure that the estimate is nearly unbiased. Our estimator requires knowing a rough estimate $\medt$ of the mean and an upper-bound $\vart$ on the variance of the random variable. We explain how to find $\medt$ nondestructively in Proposition~\ref{Prop:medi}, whereas $\vart$ will be provided directly by our applications.

\begin{algorithm}[H]
  \caption{Unbiased mean estimator, $\qestim(X,t,\medt,\vart,\eps)$.}
  \label{Alg:unEstim}
  \begin{algorithmic}[1]
    \State Define the shifted random variable $\bar{X} = X - \medt$.
    \State Set $a_{-1} = 0$ and $a_j = 2^j \vart$ for $j \in \itv{0}{k}$, where $k = \ceil{\log(580/\eps)}$.
    \For{$j = 0$ to $k$}
        \State\StateInd{Use a controlled rotation to transform $\ket{\pi_X}$ into the state
          \[\ket{\pi_j^+} = \sum_{x \notin (a_{j-1}, a_j]} \sqrt{\pi_{\bar{X}}(x)} \ket{x} \ket{0} + \sum_{x \in (a_{j-1}, a_j]} \sqrt{\pi_{\bar{X}}(x)} \ket{x}\pt*{\sqrt{1-\frac{x}{a_j}}\ket{0} + \sqrt{\frac{x}{a_j}} \ket{1}}\]
        and denote $\mu^+_j = \ex[\Big]{\frac{\bar{X}}{a_j} \ind{\bar{X} \in (a_{j-1}, a_j]}} = \norm{(\id \otimes \proj{1}) \ket{\pi_j^+}}^2$.}
        \State\StateInd{Compute an estimate $\mut^+_j$ of $\mu^+_j$ by using the \hyperref[Thm:nduAE]{nondestructive unbiased amplitude estimation} algorithm $\nduae(\ket{\pi_j^+},\id \otimes \proj{1},300t,\eps')$ where $\eps' = \pt[\big]{\frac{\eps}{2320t}}^2$.}
        \State\StateInd{Restore $\ket{\pi_j^+}$ to $\ket{\pi_X}$ by undoing the controlled rotation.}
        \State\StateInd{Repeat steps 4--6 with the state
          \[\ket{\pi_j^-} = \sum_{- x \notin (a_{j-1}, a_j]} \sqrt{\pi_{\bar{X}}(x)} \ket{x} \ket{0} + \sum_{-x \in (a_{j-1}, a_j]} \sqrt{\pi_{\bar{X}}(x)} \ket{x}\pt*{\sqrt{1-\frac{\abs{x}}{a_j}}\ket{0} + \sqrt{\frac{\abs{x}}{a_j}} \ket{1}}\]
        to obtain an estimate $\mut^-_j$ of $\mu^-_j = \ex[\Big]{\frac{\abs{\bar{X}}}{a_j} \ind{\bar{X} \in [-a_j, -a_{j-1})}}$.}
     \EndFor
    \State Output $\mut = \medt + \sum_{j=0}^k a_j(\mut^+_j - \mut^-_j)$.
  \end{algorithmic}
\end{algorithm}

\begin{proposition}[\sc Unbiased mean estimator]
  \label{Prop:unEstim}
  Let $X$ be a finite random variable with mean~$\mu$ and variance~$\sigma^2$ encoded as a qsample $\ket{\pi_X}$. Given two reals $\eps \in (0,1)$ and $t \geq 1$ and two estimates $\medt, \vart$ such that $\abs{\mu - \medt} \leq 17 \sigma$ and $\vart \geq \sigma$, the \emph{unbiased mean estimator} $\qestim\pt{X,t,\medt,\vart,\eps}$ (Algorithm~\ref{Alg:unEstim}) outputs a mean estimate~$\mut$ such that
    \[\abs*{\ex*{\mut} - \mu} \leq \eps \vart \quad \text{and} \quad \var{\mut} \leq \pt*{\frac{\vart}{t}}^2.\]
  The algorithm needs one copy of $\ket{\pi_X}$, which is restored at the end of the computation with probability at least $1-\eps$, and $\wbo{t\log(1/\eps)^2}$ applications of the reflection $\id - 2\proj{\pi_X}$ in expectation.
\end{proposition}

\begin{proof}
  We start by giving several properties of the shifted random variable $\bar{X} = X-\medt$. First, its second moment is at most
    \begin{equation}
      \label{Eq:median}
      \ex{\bar{X}^2} = \ex{(X - \mu)^2} + (\mu-\medt)^2 \leq 290 \sigma^2
    \end{equation}
 since, by assumption, the input estimate $\medt$ satisfies $\abs{\mu - \medt} \leq 17 \sigma$. Next, the expectation of $\abs{\bar{X}}$ over any interval $(a_{j-1}, a_j]$ defined in Algorithm~\ref{Alg:unEstim} when $j > 0$ is at most $\ex[\big]{\abs{\bar{X}} \ind{\abs{\bar{X}} \in (a_{j-1}, a_j]}} \leq \ex[\big]{\bar{X}^2  \ind{\abs{\bar{X}} \in (a_{j-1}, a_j]}}/a_{j-1}$ because $a_{j-1} > 0$. Hence, the normalized means $\mu^+_j$ and~$\mu^-_j$ satisfy
    \begin{equation}
      \label{Eq:mean2}
      \mu^+_j + \mu^-_j \leq \frac{2}{a_j^2} \ex[\big]{\bar{X}^2  \ind{\abs{\bar{X}} \in (a_{j-1}, a_j]}} \qquad \text{when $j > 0$.}
    \end{equation}
  Lastly, the expectation of $\abs{\bar{X}}$ over the interval $[a_k,+\infty)$ is at most
    \begin{equation}
      \label{Eq:mean3}
      \ex{\abs{\bar{X}} \ind{\abs{\bar{X}}} \geq a_k} \leq \frac{\ex{\bar{X}^2}}{a_k} \leq \frac{290\sigma}{2^k}
    \end{equation}
  where we used Equation~(\ref{Eq:median}) and the assumption that $\vart \geq \sigma$.

  We now analyze the expectation and variance of the final estimate $\mut$. Each of the $\mut^+_j$ and~$\mut^-_j$ estimates introduces a bias of at most $\eps'$ by Theorem~\ref{Thm:nduAE} and there is another bias due to the truncated part $\ex{\bar{X} \ind{\abs{\bar{X}}} \geq a_k}$ analyzed in Equation~(\ref{Eq:mean3}). Overall, we have $\abs*{\ex*{\mut} - \mu} \leq \sum_{j=0}^k a_j \eps' + \frac{290\sigma}{2^k} \leq \pt{2^{k+1} \eps' + \frac{290}{2^k}} \vart \leq \eps \vart$. Since the estimates $\mut_j$ are independent random variables, the variance is $\var{\mut} = \sum_{j=0}^{k} a_j^2 \var{\mut_j}$. Thus, by Theorem~\ref{Thm:nduAE} and Equation~(\ref{Eq:mean2}), $\var{\mut} \leq \sum_{j=0}^{k} a_j^2 \pt*{\frac{\mu^+_j + \mu^-_j}{(30t)^2} + \eps'} \leq a_0^2 \pt*{\frac{2}{(30t)^2} + \eps'} + \sum_{j=1}^{k} \pt[\Big]{\frac{2 \ex[\big]{\bar{X}^2 \ind{\abs{\bar{X}} \in (a_{j-1}, a_j]}}}{(30t)^2} + a_j^2 \eps'} \leq \frac{2\vart^2 + 2 \ex*{\bar{X}^2}}{(30t)^2} + 2^{2k+1} \vart^2 \eps' \leq \pt*{\frac{582}{(30t)^2} + 2^{2k+1} \eps'} \vart^2 \leq \pt[\big]{\frac{\vart}{t}}^2$.

  By Theorem~\ref{Thm:nduAE}, the algorithm uses $\wbo[\big]{k T \log(1/\eps')} = \wbo{t \log(1/\eps)^2}$ reflections through~$\ket{\pi_X}$ and it preserves the copy of $\ket{\pi_X}$ with probability at least $1-2(k+1)\eps' \geq 1-\eps$.
\end{proof}

We now explain how to find an estimate $\medt$ that satisfies the first requirement $\abs{\mu - \medt} \leq 17 \sigma$ needed to apply the above estimator. If we had access to classical samples then it would suffice to compute the average of $\bo{\log 1/\eta}$ samples to get the above inequality with probability $1-\eta$. Nevertheless, we do not know how to extract a classical sample from a qsample $\ket{\pi_X}$ without destroying it. Since the mean is always within one standard deviation of the median, we can instead attempt to estimate the median. There exist quantum algorithms~\cite{NW99c,Nay99d,Ham21c} for computing the quantiles of a distribution, however they also require measuring qsamples. We follow a new approach that incurs a logarithmic overhead in the size of the support but is nondestructive. Our algorithm consists of performing a binary search over the support of the random variable, where we test if a value~$x$ is close to the median by estimating the tail probability $\pr{X \geq x}$ with the nondestructive amplitude estimation algorithm.

\begin{algorithm}[H]
  \caption{Nondestructive median estimator, $\medi(X,\eta)$.}
  \label{Alg:medi}
  \begin{algorithmic}[1]
    \State Let $x_1 < \dots < x_n$ denote the values in the support of $X$.
    \State Set $a = 1$ and $b = n+1$. Run the following binary search until $a = b$:
      \Indent
        \State\StateInd{Set $k = \floor{\frac{a+b}{2}}$.}
        \State\StateInd{Compute an estimate $\td{p}_k$ of $p_k = \pr{X \geq x_k}$ by running the \hyperref[Prop:ndAE]{nondestructive amplitude estimation} algorithm $\ndae(\ket{\pi_k},\id \otimes \proj{1},3\sqrt{2},\eta/\log(m))$ where
        \[\ket{\pi_k} = \pt[\Bigg]{\sum_{x < x_k} \sqrt{\pi_X(x)} \ket{x}} \ket{0} + \pt[\Bigg]{\sum_{x \geq x_k} \sqrt{\pi_X(x)} \ket{x}}\ket{1}.\]}
        \State\StateInd{If $\td{p}_k \leq 1/6$ then set $b = k$ else set $a = k+1$.}
      \EndIndent
    \State Output $\medt = x_{a-1}$.
  \end{algorithmic}
\end{algorithm}

\begin{proposition}[\sc Nondestructive median estimator]
  \label{Prop:medi}
  Let $X$ be a finite random variable with support size $n$, mean $\mu$ and variance $\sigma^2$. Given a real $\eta \in (0,1)$, the \emph{nondestructive median estimator} $\medi(X,\eta)$ (Algorithm~\ref{Alg:medi}) outputs an approximate median~$\medt$ such that
    \[\abs{\mu - \medt} \leq 17\sigma\]
  with probability at least $1-\eta$.
  The algorithm needs one copy of $\ket{\pi_X}$, which is restored at the end of the computation with probability at least $1-\eta$, and $\bo{\log(n)\log(\log(n)/\eta)}$ applications of the reflection $\id - 2\proj{\pi_X}$.
\end{proposition}

\begin{proof}
  By Theorem~\ref{Prop:ndAE} and a union-bound argument, all the estimates $\td{p}_k$ computed by Algorithm~\ref{Alg:medi} satisfy $\abs{\td{p}_k - p_k} \leq \sqrt{p_k}/(3\sqrt{2}) + 1/18$ with probability at least $1-\eta$. In particular, if $p_k \geq 1/2$ then $\td{p}_k \geq p_k - \sqrt{p_k}/(3\sqrt{2}) - 1/18 \geq 5/18$ and if $p_k \leq 1/18$ then $\td{p}_k \leq p_k + \sqrt{p_k}/(3\sqrt{2}) + 1/18 \leq 1/6$. Thus, the binary search stops with $a = b$ satisfying $p_a < 1/2 \leq 9p_{a-1}$ with probability at least $1-\eta$ (note that $a \geq 2$ since $p_1 = 1$). Consequently, if we define the quantile function $Q_X(p) = \sup\set{x \in \R : \pr{X \geq x} \geq p}$ for all $p \in [0,1]$, we have $Q_X(1/2) \leq \medt \leq Q_X(1/18)$. Finally, the result $\abs{\mu - \medt} \leq 17\sigma$ follows from the quantile inequality $\mu - \sigma \sqrt{\frac{p}{1-p}} \leq Q_X(p) \leq \mu + \sigma \sqrt{\frac{1-p}{p}}$ proved in~\cite[Section 3.1]{BB04j}.
\end{proof}

We describe our final estimator for estimating the mean of a product of $\ell$ independent random variables $X_1,\dots,X_{\ell}$ with relative error $\eps$. We make the standard assumptions that the relative second moment $\frac{\ex{X_i^2}}{\ex{X_i}^2}$ and the inverse fidelity $1/\abs{\ip{\pi_{X_i}}{\pi_{X_{i+1}}}}^2$ are upper-bounded by some known value~$B$, which is a constant in our applications (Section~\ref{Sec:partition}). Our estimator uses a similar approach to that of Dyer and Frieze~\cite{DF91c} by taking the product of $\ell$ (nearly) unbiased estimates $\mut_1, \dots, \mut_{\ell}$ of $\ex{X_1},\dots,\ex{X_{\ell}}$ respectively, each with variance $\bo{\eps^2/\ell}$. We obtain a quantum speed-up by using the quantum unbiased estimator of Proposition~\ref{Prop:unEstim} to reduce the variances quadratically faster than classically.

\begin{algorithm}[H]
  \caption{Product estimator, $\qprod(X_1,\dots,X_{\ell},B,\eps)$.}
  \label{Alg:prodEst}
  \begin{algorithmic}[1]
    \For{$i = 1$ to $\ell$}
      \State Define $\ha{X}_i$ to be the average of $K = \max\set{1,1156(B-1)}$ independent samples from $X_i$.
      \State\StateInd{Compute an estimate $\medt_i$ of the median of $\ha{X}_i$ by using the \hyperref[Prop:medi]{nondestructive median estimator} $\medi(\ha{X}_i,1/(11\ell))$. Set $\vart_i = \medt_i B$.}
      \State\StateInd{Compute an estimate $\mut_i$ of the expectation $\ex{\ha{X}_i}$ by using the \hyperref[Prop:unEstim]{unbiased mean estimator} $\qestim(\ha{X}_i,\frac{96B\sqrt{\ell}}{\eps},\medt_i,\vart_i,\frac{\eps}{6\ell B})$.}
      \State Anneal the qsample $\ket{\pi_{\ha{X}_i}}$ to $\ket{\pi_{\ha{X}_{i+1}}}$ with \hyperref[Lem:ndAnneal]{probabilistic annealing}.
    \EndFor
    \State Output $\mut = \mut_1 \cdots \mut_{\ell}$.
  \end{algorithmic}
\end{algorithm}

\begin{theorem}[\sc Product estimator]
  \label{Thm:prodEst}
  Let $B > 1$ and $\eps \in (0,1)$. Consider a sequence $X_1,\dots,X_{\ell}$ of $\ell$ independent random variables with support size $n$, bounded relative second moment $\frac{\ex{X_i^2}}{\ex{X_i}^2} \leq B$ and bounded fidelity $\abs{\ip{\pi_{X_i}}{\pi_{X_{i+1}}}}^2 \geq 1/B$ for all $i$. Denote their product as $X = X_1 \cdots X_{\ell}$. Then, the \emph{product estimator} $\qprod(X_1,\dots,X_{\ell},B,\eps)$ (Algorithm~\ref{Alg:prodEst}) outputs a multiplicative-error estimate~$\mut$ such that
    \[\abs[\Big]{\mut - \ex[\Big]{\prod_{i=1}^{\ell} X_i}} \leq \eps \ex[\Big]{\prod_{i=1}^{\ell} X_i}\]
  with probability at least~$2/3$. It uses $\bo{B}$ copies of $\ket{\pi_{X_1}}$ and $\wbo[\big]{B^2\ell^{3/2}/\eps + B \ell \log(n)}$ reflections through the states $\ket{\pi_{X_1}},\dots,\ket{\pi_{X_{\ell}}}$ in expectation.
\end{theorem}

\begin{proof}
  By a union-bound argument, with probability at least $10/11$, all the median estimates $\medt_i$ computed at step 3 are within distance $17\sqrt{\var{\ha{X}_i}}$ of the means $\ex{\ha{X}_i}$. We condition on this event happening for the rest of the proof.

  We start by analysing steps 2--5 for a given random variable $X_i$. First, notice that we can construct a qsample for $\ha{X}_i$ by taking $\ket{\pi_{\ha{X}_i}} = U_{\mathrm{average}}\pt{\ket{\pi_{X_i}}^{\otimes K} \otimes \ket{0}}$ where $U_{\mathrm{average}}$ is a unitary computing the average of the first $K$ registers in the last register. The expectation is unchanged $\ex{X_i} = \ex{\ha{X}_i}$, whereas the relative variance is at most
    \begin{equation}
      \label{Eq:qsampleAve}
      \frac{\var{\ha{X}_i}}{\ex{X_i}^2} = \frac{\var{X_i}}{K\ex{X_i}^2} \leq \frac{B-1}{K}.
    \end{equation}
  The next equation shows that the distance between the rough estimate $\medt_i$ computed at step~3 and the expectation is bounded by a multiple of the standard deviation (by Proposition~\ref{Prop:medi}), which in turn is a small multiple of the expectation (since the relative variance is small by Equation~(\ref{Eq:qsampleAve})).
    \begin{equation}
      \label{Eq:mean4}
      \abs{\medt_i - \ex{X_i}} \leq 17 \sqrt{\var{\ha{X}_i}} \leq 17\sqrt{\frac{B-1}{K}} \ex{X_i} = \frac{\ex{X_i}}{2}
    \end{equation}
  Finally, the expectation of each estimate $\mut_i$ computed at step 4 is bounded as
    \begin{equation}
      \label{Eq:mean5}
      \abs{\ex{\mut_i} - \ex{X_i}} \leq \frac{\eps}{6\ell B} \vart_i \leq \frac{\eps}{4\ell}\ex{X_i}
    \end{equation}
  where the first inequality is by Theorem~\ref{Thm:nduAE} and the second is by $\vart_i = B \medt_i$ and Equation~(\ref{Eq:mean4}).

  Consider now the product random variable $X = X_1 \cdots X_{\ell}$. The expectation of the final estimate~$\mut$ is at most $\ex{\mut} \leq (1+ \frac{\eps}{4\ell})^{\ell} \prod_{i = 1}^{\ell} \ex{\mu_i} \leq (\eps/2) \ex{X}$ and at least  $\ex{\mut} \geq (1- \frac{\eps}{4\ell})^{\ell} \prod_{i = 1}^{\ell} \ex{\mu_i} \geq (\eps/4) \ex{X}$ by Equation~(\ref{Eq:mean5}). Thus, $\abs{\ex{\mut} - \ex{X}} \leq (\eps/2) \ex{X}$.  The relative variance of $\mut$ is at most $\frac{\var{\mut}}{\ex{\mut}^2} = \prod_{i = 1}^{\ell} \pt*{1 + \frac{\var{\mut_i}}{\ex{\mut_i}^2}} - 1 \leq \prod_{i = 1}^{\ell} \pt*{1 + \frac{\vart_i^2}{(t\ex{\mut_i})^2}} - 1 = \prod_{i = 1}^{\ell} \pt[\Big]{1 + \pt*{\frac{\eps B \medt_i}{96\sqrt{\ell} B \ex{\mut_i}}}^2} - 1 \leq \pt*{1 + \frac{\eps^2}{32\ell}}^{\ell} - 1 \leq \eps^2/16$ where the second step is by Proposition~\ref{Prop:unEstim} and the fourth step uses that $\medt_i/\ex{\mut_i} \leq 3$ by Equations~(\ref{Eq:mean4}) and~(\ref{Eq:mean5}). We conclude by using the triangle and Chebyshev inequalities that $\abs{\mut - \ex{X}} \leq \abs{\mut - \ex{\mut}} + \abs{\ex{\mut} - \ex{X}} \leq (\eps/2 + \eps/2) \ex{X}$ with probability at least $1 - 4\frac{\var{\mut}}{(\eps\ex{\mut})^2} \geq 3/4$.

  The algorithm needs $K$ copies of the initial qsample $\ket{\pi_{X_1}}$ to anneal successively into~$K$ copies of~$\ket{\pi_{X_i}}$ at each stage of steps~2--5. The transition from $\ket{\pi_{X_i}}^{\otimes K}$ to $\ket{\pi_{X_{i+1}}}^{\otimes K}$ requires $\bo{K B}$ reflections through $\ket{\pi_{X_i}}$ and $\ket{\pi_{X_{i+1}}}$, by using Lemma~\ref{Lem:ndAnneal} on each of the~$K$ subsystems. Furthermore, at each stage, step~3 uses $\wbo{\log(n)\log(\ell)}$ reflections through~$\ket{\pi_{\ha{X}_i}}$ by Proposition~\ref{Prop:medi} and step~4 uses $\wbo{B\sqrt{\ell}/\eps}$ reflections through $\ket{\pi_{\ha{X}_i}}$ by Proposition~\ref{Prop:unEstim}. Overall, the number of reflections through the original states $\ket{\pi_{X_1}},\dots,\ket{\pi_{X_{\ell}}}$ is thus $\wbo[\big]{\ell K \pt[\big]{B + \log(n) + B\sqrt{\ell}/\eps}}$.
\end{proof}

%% file: Sources/partition_intro.tex
In this section, we showcase an application of our newly-constructed quantum mean estimation algorithm. Specifically, we show how it can be used to speed up existing quantum algorithms for estimating partition functions. In Section~\ref{Sec:annealing}, we elaborate on the generic algorithm for partition function estimation and how our results provide a speed up and, in Section~\ref{Sec:applications}, we discuss how this gives rise to more efficient quantum algorithms for several applications.

%% file: Sources/annealing.tex
In this subsection, we first describe our partition function estimation algorithm in high level, and define the required notation along the way. We follow the exposition in \cite{HW20c}, but use slightly different notations to match the rest of this document more closely.

Let $\Omega$ be a state space, and let $H : \Omega \to \itv{0}{n}$ be a Hamiltonian, i.e., a function associating an energy to each of the states. Our goal will be to compute the number of states whose energy is $0$, i.e., $|H^{-1}(0)|$. To that end, we define the partition function $Z : [0,\infty] \to \mathbb{R}$ as
\[Z(\beta) = \sum_{x \in \Omega} e^{-\beta H(x)},\]
and we observe that $Z(0) = |\Omega|$, and $Z(\infty) = |H^{(-1)}(0)|$. The partition function arises frequently in statistical physics, where $\beta$ is referred to as the inverse temperature.

Next, we define a sequence of inverse temperatures $0 = \beta_0 < \beta_1 < \cdots < \beta_{\ell}= \infty$, and express~$Z(\infty)$ as the telescoping product
\begin{equation}
	\label{eq:telescoping-product}
	Z(\infty) = Z(0) \cdot \frac{Z(\beta_1)}{Z(0)} \cdot \frac{Z(\beta_2)}{Z(\beta_1)} \cdot \raisebox{-.2em}{$\cdots$} \cdot \frac{Z(\infty)}{Z(\beta_{\ell})} = Z(0) \prod_{i=0}^{\ell-1} \frac{Z(\beta_{i+1})}{Z(\beta_i)}.
\end{equation}
Since throughout the sequence of $\beta_i$'s, we are increasing the inverse temperature, and hence decreasing the temperature, this sequence is referred to as the \textit{cooling schedule}. Its length $\ell$ is called the \textit{schedule length}. The core idea for estimating $Z(\infty)$ is to evaluate each of the factors in the product on the right-hand side of Equation~(\ref{eq:telescoping-product}) individually.

Thus, let $i \in [\ell]$. We endow $\Omega$ with a probability distribution $\pi_{\beta_i}$, called the Gibbs distribution in statistical physics, and define a random variable $X_i : \Omega \to \mathbb{R}$ on it, with
  \[\pi_{\beta_i}(x) = \frac{e^{-\beta_iH(x)}}{Z(\beta_i)}, \qquad \text{and} \qquad X_i(x) = e^{-(\beta_{i+1} - \beta_i)H(x)}.\]
It follows immediately that
  \[\ex{X_i} = \sum_{x \in \Omega} \frac{e^{-\beta_i H(x)}}{Z(\beta_i)} \cdot e^{-(\beta_{i+1} - \beta_i)H(x)} = \sum_{x \in \Omega} \frac{e^{-\beta_{i+1} H(x)}}{Z(\beta_i)} = \frac{Z(\beta_{i+1})}{Z(\beta_i)}.\]
The idea is now to devise a procedure that samples from the Gibbs distribution and obtains an estimator $\mut_i$ of~$\ex{X_i}$.

In order to be able to sample from the Gibbs distribution, we encode it in a quantum state. To that end, we define the Gibbs state $\ket{\pi_{\beta_i}} \in \C^{\Omega}$ at inverse temperature $\beta_i$, defined as
\[\ket{\pi_{\beta_i}} = \frac{1}{\sqrt{Z(\beta_i)}} \sum_{x \in \Omega} \sqrt{e^{-\beta_iH(x)}}\ket{x}.\]
In typical applications (see the next subsection), it can be much easier to reflect through the Gibbs state than to prepare it. Therefore, we estimate $\ex{X_i}$ using the nondestructive mean estimation algorithm constructed in the previous sections. This has the benefit of restoring all the Gibbs state needed in the computation of $\mut_i$, which makes it possible to reuse them for computing the next factor in the products on the right-hand side of Equation~(\ref{eq:telescoping-product}).

It remains to choose the cooling schedule, i.e., the inverse temperatures $\beta_1, \dots, \beta_{\ell-1}$. For reasons that are sketched on high level, we want our cooling schedule to have the following two properties.
\begin{enumerate}
	\item \textit{$B$-Chebyshev}. For any $B \geq 1$, a cooling schedule is called $B$-Chebyshev if for any two subsequent $\beta_i$ and $\beta_{i+1}$, we have
	 \[\frac{\ex{X_i^2}}{\ex{X_i}^2} = \frac{Z(2\beta_{i+1} - \beta_i)Z(\beta_i)}{Z(\beta_{i+1})^2} \leq B.\]
	This requirement can be viewed as a tail bound -- it tells us that $X_i$ concentrates well around its mean~$\ex{X_i}$.
	\item \textit{$B$-slowly varying}. For any $B \geq 1$, a cooling schedule is called $B$-slowly varying if
	\[\abs[\big]{\ip{\pi_{\beta_i}}{\pi_{\beta_{i+1}}}}^2 = \frac{Z\pt[\big]{\frac{\beta_i + \beta_{i+1}}{2}}^2}{Z(\beta_i)Z(\beta_{i+1})} \geq 1/B.\]
	This requirement can be seen as a proximity (or \emph{warm start}) requirement -- it tells us that the Gibbs states $\ket{\pi_{\beta_i}}$ and $\ket{\pi_{\beta_{i+1}}}$ have some non-negligible overlap. It ensures that we can \textit{anneal}, i.e., transform, $\ket{\pi_{\beta_i}}$ into $\ket{\pi_{\beta_{i+1}}}$ without having to do too much work, and hence circumvents having to prepare $\ket{\pi_{\beta_{i+1}}}$ from scratch for estimating the next mean $\ex{X_{i+1}}$.
\end{enumerate}
When we set $B = e^2$, it is shown in \cite{SVV09j} that there exists a cooling schedule of length $\ell = \wbo{\sqrt{\log|\Omega|\log(n)}}$, which is $B$-Chebyshev. Subsequently, it was shown in \cite{HW20c} that one can modify the construction so that the resulting cooling schedule is not only $B$-Chebyshev, but also slowly varying. Moreover, computing this cooling schedule can be done on the fly, with associated cost scaling approximately linearly in the schedule length. Combining this work with our product estimator, Theorem~\ref{Thm:prodEst}, gives rise to the following result.

\begin{theorem}[\sc Partition function estimator]
  \label{Thm:QSA}
  Let $H : \Omega \to \itv{0}{n}$ be a Hamiltonian, and let $Z$ be its partition function. Suppose that, for every inverse temperature~$\beta$, the associated Gibbs distribution $\pi_{\beta}$ is the stationary distribution of an ergodic reversible Markov chain with spectral gap at least~$\delta$. Then, we can estimate~$Z(\infty)$ up to multiplicative error $\eps$, with probability at least~$2/3$, by using
    \[\wbo[\Big]{\pt[\big]{\log^{3/4}|\Omega|\log^{3/4}(n)/\eps + \sqrt{\log|\Omega|}\log^{3/2} n}/\sqrt{\delta}}\]
  steps of the quantum walk operator in expectation.
\end{theorem}

\begin{proof}
    We run the algorithm from \cite{HW20c} to compute the cooling schedule on the fly. The total number of Gibbs state reflections used in this step is $\wbo{\sqrt{\log|\Omega|}\log^{3/2} n}$~\cite[Theorem~13]{HW20c}, and the resulting cooling schedule is both $e^2$-slowly varying and $e^2$-Chebyshev, and of length $\ell = \wbo{\sqrt{\log|\Omega|\log n}}$. For computing the product in Equation~(\ref{eq:telescoping-product}), we use the product estimator from Theorem~\ref{Thm:prodEst}. The total number of Gibbs state reflections used in this algorithm is $\wbo{\log^{3/4}|\Omega|\log^{3/4}(n)/\eps + \sqrt{\log|\Omega|}\log^{3/2} n}$. With standard failure probability reduction techniques, we obtain that the number of Gibbs state reflections performed by the resulting algorithm scales as the sum of the two complexities. Finally, by a well-known result~\cite{Sze04c,MNRS11j}, each reflection through a Gibbs state~$\ket{\pi_{\beta}}$ can be implemented with $\bo{1/\sqrt{\delta}}$ steps of the quantum walk operator corresponding to a Markov chain with spectral gap at least $\delta$ generating~$\pi_{\beta}$.
\end{proof}

The core of our improvement lies in our product estimator -- the one used in \cite{HW20c} is quadratic in the schedule length, whereas ours is subquadratic. Thus, more generally, if we are in a setting where a cooling schedule that is shorter than $\wbo{\sqrt{\log|\Omega|\log(n)}}$ suffices, then our resulting algorithm scales with the schedule length to the power of $3/2$. A typical setting where this is the case is when we know a priori a lower bound on $Z(\infty)$ -- then we can terminate the annealing at a lower inverse temperature, and hence obtain a shorter cooling schedule. We leave working out the details in this setting for future work.

%% file: Sources/applications.tex
The quantum partition function estimator can be applied to several combinatorial counting and statistical physics problems. We describe some examples that are representative of these applications and for which we obtain faster algorithms than in previous work~\cite{Mon15j,CCH+19p,HW20c,AHN22j}. In all results that we list below, the range~$n$ of the Hamiltonian scales at most polylogarithmically in $|\Omega|$, so we can hide any dependence on $\log(n)$ in the tilde of the big-$O$ notation.

\paragraph{Counting $k$-colorings.} Let $G = (V,E)$ be a graph of maximum degree~$\Delta = \bo{1}$, and a number $k = \bo{1}$ of colors such that $k > 2\Delta$. We want to count the number of ways to color the vertices of $G$ such that no edge is monochromatic. We let $\Omega$ be the set of all colorings~$x \in [k]^V$, which implies that~$Z(0) = |\Omega| = k^{|V|}$, and $H(x)$ is the number of monochromatic edges in~$G$ when each vertex $v \in V$ is colored with $x_v$. This is also called the \emph{Potts model} and $Z(\infty) = |H^{(-1)}(0)|$ is equal to the number of valid $k$-colorings of $G$. Jerrum \cite{Jer95j} showed that the Glauber dynamics for the Potts model mixes in time~$O(|V|\log|V|)$. Thus, by Theorem~\ref{Thm:QSA}, we can estimate~$Z(\infty)$ up to relative error~$\eps$ in time $\wbo[\big]{\log^{3/4}(|\Omega|)\sqrt{|V|}/\eps} = \wbo{|V|^{5/4}/\eps}$.

\paragraph{Ferromagnetic Ising model.} Let $G = (V,E)$ be a graph of maximum degree~$\Delta \geq 3$. We take~$\Omega$ to be the set of all assignments~$x \in \set{-1,1}^V$ of signs~$\pm 1$ to the vertices of~$G$, which readily implies that~$|\Omega| = 2^{|V|}$. The energy~$H(x)$ is the number of edges in $E$ whose two endpoints have the same sign (when the sign of~$v \in V$ is~$x_v$). Here, since the model is ferromagnetic, the Gibbs distribution is chosen to be proportional to~$e^{\beta H(x)}$ and the partition function is $Z(\beta) = \sum_x e^{\beta H(x)}$ (the framework given in the previous section can easily be adapted to this setting). Mossel and Sly~\cite{MS13j} showed that the Glauber dynamics for the ferromagnetic Ising model mixes in time~$O(|V| \log |V|)$ in the tree uniqueness region $\beta < \ln\pt{\Delta/(\Delta-2)}$. Thus, for inverse temperatures~$\beta$ satisfying that condition, we can estimate~$Z(\beta)$ in time $\wbo{|V|^{5/4}/\eps}$.

\paragraph{Counting matchings.} Let $G = (V,E)$ be a graph of maximum degree~$\Delta = \bo{1}$. A matching $x \subseteq E$ is a subset of disjoint edges. We let $\Omega$ be the set of all matchings, and we set $H(x) = \abs{x}$ to be the number of edges in the matching $x$. This is called the \emph{monomer-dimer model}. The state space is of size $Z(0) = |\Omega| = O(|V|! 2^{|V|})$ and we have~$Z(\infty) = 1$. The setting is again slightly different here since the partition function is easy to calculate at zero temperature ($\beta = \infty$) rather than at infinite temperature ($\beta = 0$). We can still run a similar algorithm as in Theorem~\ref{Thm:QSA} by annealing in the opposite direction (see~\cite{HW20c,Mon15j} for a more detailed explanation). Chen, Liu and Vigoda~\cite[Theorem 1.5]{CLV21c} showed that the Glauber dynamics for the monomer-dimer model mixes in time $\bo{|E| \log |V|}$. Thus, we can estimate the number~$Z(0)$ of matchings in $G$ in time $\wbo{|V|^{3/4} |E|^{1/2}/\eps}$.

\paragraph{Counting independent sets.} Let $G = (V,E)$ be a graph of maximum degree~$\Delta = \bo{1}$. An independent set $x \subseteq V$ is a subset of vertices such that no edge has its two endpoints in~$x$. Similarly as before, we let $\Omega$ be the set of all independent sets and we define~$H(x) = \abs{x}$ to be the size (number of vertices) of the independent set $x$. This is also called the \emph{hard-core model}. Again, we have $Z(\infty) = 1$ and we anneal backwards to estimate $Z(0) = |\Omega|$.
Chen, Liu and Vigoda~\cite[Theorem 1.2]{CLV21c} showed that the Glauber dynamics for the hard-core model mixes in time $\bo{|V| \log |V|}$ in the tree uniqueness region $\beta > \ln\pt[\Big]{\frac{(\Delta-2)^{\Delta}}{(\Delta-1)^{\Delta-1}}}$. Thus, for inverse temperatures $\beta$ satisfying that condition, we can estimate $Z(\beta)$ in time $\wbo{|V|^{5/4}/\eps}$. In particular, we can estimate the number of independent sets when the maximum degree is $\Delta \leq 5$ (since the tree uniqueness region contains the value $\beta = 0$ in that case).

\paragraph{Computing the volume of a convex body.} We can also use our techniques to speed up the volume estimation algorithms of convex bodies from~\cite{LV06j,CCH+19p}. In this problem, we assume that we are given a radius $R > 0$ and a convex set $K \subseteq \R^d$ such that $B(1) \subseteq K \subseteq B(R)$ where~$B(r)$ is the ball of radius $r$ centered at $0 \in \R^d$. The goal is to compute an estimate $\td{V}$ such that $(1-\eps)\Vol(K) \leq \td{V} \leq (1+\eps)\Vol(K)$. The algorithm from~\cite{CCH+19p} achieves this with~$\wbo{d^3 + d^{2.5}/\eps}$ quantum queries to a membership oracle~$O_K$ to $K$, and we reduce the number of queries required to $\wbo{d^3 + d^{2.25}/\eps}$. In contrast to the other applications mentioned in this section, here we require a more significant change to the algorithm, and hence we supply the details in Appendix~\ref{App:volumeEstimation}.

%% file: Sources/acknowledgements.tex
The authors thank András Gilyén for discussions on unbiased phase estimation. Y.H. was supported by the Simons Institute and DOE NQISRC QSA grant \#FP00010905.

%% file: Sources/app_primitives.tex
We use the next properties that characterize the output distribution of the quantum phase estimation algorithm.

\begin{lemma}[\sc Phase estimation~\cite{Kit95p,BHMT02j}]
  \label{Lem:PE}
  Let $U$ be a unitary operator with an eigenvector~$\ket{\psi}$ such that $U \ket{\psi} = e^{2\pi i \theta} \ket{\psi}$ where $\theta \in \bp{0,1}$. Given a power-of-two integer $t$, the \emph{phase estimation} algorithm outputs an estimate $\tht = i/t$ where the number $i \in \itv{0}{t-1}$ is chosen with probability
    \begin{equation}
      \label{Eq:PE}
      p(i) = \frac{\sin^2(t \Delta_i \pi)}{t^2\sin^2(\Delta_i \pi)}
    \end{equation}
  and $\Delta_i = \min\set*{\abs{\vp - i/t},\abs{1 + \vp - i/t},\abs{1 -\vp + i/t}}$ is the distance between $\vp$ and $i/t$ along the unit circle. The algorithm needs one copy of $\ket{\psi}$, which is restored at the end of the computation, and $\bo{t}$ controlled-$U$ operations.
\end{lemma}

\begin{corollary}
  \label{Cor:phaseProba}
  If $t \geq 8$ then the probability distribution $(p(i))_i$ defined in Equation~(\ref{Eq:PE}) satisfies
  $p(i) \geq 0.22$ when $\Delta_i \leq 5/(8t)$ and $p(i) \leq 0.11$ when $\Delta_i \geq 1/t$.
\end{corollary}

\begin{proof}
  Suppose first that $\Delta_i \leq 5/(8t)$. Using the inequality $\sin(x) \leq x$ when $0 \leq x \leq \pi/2$ and the fact that $y \mapsto (\sin(y)/y)^2$ is non-increasing on $[0,5\pi/8]$, we have $p(i) \geq (\sin(t \Delta_i \pi)/(t \Delta_i \pi))^2 \geq (\sin(5\pi/8)/(5\pi/8))^2 \geq 0.22$. Suppose now that $\Delta_i \geq 1/t$. Using that $\Delta_i \leq 1/2$ and $t \geq 8$, we have $p(i) \leq 1/(t \sin(\pi/t))^2 \leq 1/(8 \sin(\pi/8))^2 \leq 0.11$.
\end{proof}

We consider a variant of quantum amplitude amplification~\cite{BHMT02j} that provides a precise linear amplification of the amplitude.

\begin{lemma}[{\sc Linear amplitude amplification}, Theorem 6.10 in \cite{Low17d}]
  \label{Lem:AA-SVT}
  Consider two reals $t \geq 1$ and $\eps \in (0,1)$. Let $\ket{\psi}$ be a quantum state and~$\Pi$ be a projection operator such that $\norm{\Pi \ket{\psi}} \leq 1/(2t)$ if $t > 1$. Then, there is a unitary operator $V_{t,\eps}$ such that
    \[\abs[\big]{\norm{\Pi V_{t,\eps} \ket{\psi}} - t \norm{\Pi \ket{\psi}}} \leq \eps\]
  which can be implemented with $\bo{t \log(1/\eps)}$ applications of the reflection operators $\id - 2\proj{\psi}$ and $\id - 2\Pi$.
\end{lemma}

The next two results provide efficient quantum algorithms for converting between phase and amplitude encodings of a real parameter.

\begin{lemma}[{\sc Amplitude-to-Phase conversion}, Theorem 14 in \cite{GAW17p}]
  \label{Lem:PO}
  Let $\ket{\psi}$ be a quantum state and~$\Pi$ be a projection operator with $p = \norm{\Pi \ket{\psi}}^2$. Then, given a real $\eps \in (0,1)$, there is a unitary operator $\pora_{\eps}$ and an integer $a = \bo{\log\log 1/\eps}$ such that
    \[\pora_{\eps} \pt[\big]{\ket{\psi}\ket{0}^{\otimes a}} = \ket{\psi} \ket{\varphi_p} \quad \text{where} \quad \norm{\ket{\varphi_p} - e^{i p}\ket{0}^{\otimes a}} \leq \eps\]
  and $\pora_{\eps}$ can be implemented with $\bo{\log(1/\eps)}$ applications of the reflection operators $\id - 2\proj{\psi}$ and $\id - 2\Pi$.
\end{lemma}

\begin{lemma}[{\sc Phase-to-Amplitude conversion}, Lemma~16 in \cite{GAW17p}]
  \label{Lem:PrO}
  Let $\pora$ be a unitary operator with an eigenvector~$\ket{\psi}$ such that $\pora \ket{\psi} = e^{i \theta} \ket{\psi}$ where $\theta \in \bc{1/4,3/4}$. Then, given a real $\eps \in (0,1)$, there is a unitary operator $V_{\eps}$ and an integer $a = \bo{\log 1/\eps}$ such that
    \[V_{\eps} \pt[\big]{\ket{\psi}\ket{0}^{\otimes a}} = \ket{\psi} \pt[\big]{\sqrt{1-p'}\ket{0}^{\otimes a} + \sqrt{p'}\ket{\Phi^{\perp}}} \quad \text{where} \quad \text{$\abs{\sqrt{p'} - \sqrt{p}} \leq \eps$ and $\ip{0}{\Phi^{\perp}} = 0$}\]
  and $V_{\eps}$ can be implemented with $\bo{\log(1/\eps)}$ applications of the (controlled) $\pora$ and $\pora^{\dag}$ operators.
\end{lemma}

Finally, we explain how to anneal from one state to the next, given that there is a non-negligible overlap between the two. The idea is due to Marriott and Watrous~\cite{MW05j} and is similar to the coin flipping algorithm (Proposition~\ref{Prop:coin}).

\begin{algorithm}[H]
	\caption{Annealing from $\ket{\psi}$ to $\ket{\phi}$, $\anneal(\ket{\psi},\ket{\phi})$}
	\label{Alg:annealing}
	\begin{algorithmic}[1]
		\State Start with the state $\ket{\psi}$.
		\Repeat
			\State Measure in the $\{\ket{\psi}\bra{\psi}, \id - \ket{\psi}\bra{\psi}\}$ basis.
			\State Measure in the $\{\ket{\phi}\bra{\phi}, \id - \ket{\phi}\bra{\phi}\}$ basis.
		\Until{the final measurement outcome is $\ket{\phi}$.}
	\end{algorithmic}
\end{algorithm}

\begin{lemma}[\sc Probabilistic annealing]
  \label{Lem:ndAnneal}
  Let $\ket{\psi}$ and $\ket{\phi}$ be two linearly independent quantum states of equal dimension. Then the expected number of reflections through $\ket{\psi}$ and $\ket{\phi}$ performed by Algorithm~\ref{Alg:annealing} is $1 + 1/(2|\ip{\psi}{\phi}|^2)$.
\end{lemma}

\begin{proof}
	Let $p = |\ip{\psi}{\phi}|^2$. Since the first iteration starts in state $\ket{\psi}$, the outcome of the first measurement is deterministic, and hence we will terminate with probability $p$. Afterward, every iteration starts in a state $\ket{\phi^{\perp}}$, which is orthogonal to $\ket{\phi}$ and in the 2D-subspace spanned by $\ket{\psi}$ and $\ket{\phi}$. Thus, the algorithm terminates with probability $f = 2p(1-p)$, in this iteration. Putting everything together, let $T$ be the number of iterations the algorithm uses. Then, $\ex{T-1|T\geq2} = \sum_{t=2}^{\infty} (t-1)f(1-f)^{t-2} = 1/f$, and so $\ex{T} = 1 + (1-p)/f = 1 + 1/(2p)$.
\end{proof}

%% file: Sources/app_nondestrAE.tex
The \emph{nondestructive amplitude estimation} algorithm is a variant of \emph{amplitude estimation} introduced by Harrow and Wei~\cite{HW20c} for restoring the initial quantum state with high probability. This is a crucial ingredient in their work for computing a cooling schedule faster than classically. In Section~\ref{Sec:unbiasedAE}, we further extend the properties of amplitude estimation by making the output nearly unbiased. Our algorithm requires a rough amplitude estimate to start with, which can be obtained from the (biased) procedure of~\cite{HW20c}.

Below, we recall the statement of the nondestructive amplitude estimation theorem of Harrow and Wei~\cite{HW20c} and we suggest a new (still biased) construction achieving this result. While the original algorithm of Harrow and Wei may be correct, its analysis seems to use two properties of amplitude estimation~\cite{BHMT02j} that do not hold in general: the collapsing of the post-measurement state to an eigenvector of the Grover operator (it can be a superposition of two eigenvectors when the amplitude $p$ is close to $0$), and the output distribution giving overwhelming probability to a single amplitude estimate (there are two high-probability estimates when the underlying phase $\frac{1}{\pi}\arcsin(\sqrt{p})$ is at equal distance from the two nearest phase estimates on the unit circle). We avoid such complications by describing a new nondestructive amplitude estimation algorithm based on the next ``uncomputation trick'' with boosted success probability.

\begin{lemma}[\sc Amplified uncomputation]
  \label{Lem:uncomp}
  Consider a quantum algorithm that consists of applying a unitary $U$ on a state $\ket{\psi}\ket{0}^{\otimes k}$ and measuring the last $k$ qubits in the standard basis. Let $\pi(1),\dots,\pi(2^k)$ denotes the probabilities of measuring the outcomes $1,\dots,2^k$ respectively. Then, given a lower bound $\lambda \leq \sum_{i=1}^{2^k} \pi(i)^2$ and a success parameter $\eta \in (0,1)$, one can sample $i \in \set{1,\dots,2^k}$ from the distribution
    \[i \sim \frac{\pi(i)^2}{\sum_{j=1}^{2^k} \pi(j)^2}\]
  and restore the state $\ket{\psi}$ with probability at least $1-\eta$ by using $\bo[\big]{\log(1/\eta)/\sqrt{\lambda}}$ applications of~$U$, $U^{\dagger}$ and $\id - 2\proj{\psi}$.
\end{lemma}

\begin{proof}
  Let $U \pt[\big]{\ket{\psi}\ket{0}^{\otimes k}} = \sum_i \alpha_i \ket{\psi_i}\ket{i}$ denote the state computed by the unitary $U$. Define the ``uncomputation'' unitary $U_{\mathrm{uncomp}} = (U^{\dagger} \otimes \id_{\C^{2^k}}) (\id_{\Hil} \otimes \cnot) (U \otimes \id_{\C^{2^k}})$ that consists of running $U$, copying the output register into a new register and running the inverse $U^{\dagger}$. Then, the amplitude of $\ket{\psi}\ket{0}\ket{i}$ in the state obtained by applying $U_{\mathrm{uncomp}}$ on $\ket{\psi}\ket{0}^{\otimes k}\ket{0}^{\otimes k}$ is
    \[(\bra{\psi}\bra{0} \bra{i}) U_{\mathrm{uncomp}} (\ket{\psi}\ket{0}\ket{0}) = \pt[\Big]{\sum_j \alpha^*_j \bra{\psi_j}\bra{j}\bra{i}} \pt[\Big]{\sum_j \alpha_j \ket{\psi_j}\ket{j}\ket{j}} = \abs{\alpha_i}^2 = \pi(i).\]
  We run fixed-point amplitude amplification~\cite{YLC14j} on the projector $\Pi = \proj{\psi} \otimes \proj{0} \otimes \id$ and the state $U_{\mathrm{uncomp}}\pt{\ket{\psi}\ket{0}\ket{0}}$, using the amplitude lower bound $\norm{\Pi U_{\mathrm{uncomp}}\pt{\ket{\psi}\ket{0}\ket{0}}}^2 \geq \lambda$ and the amplification parameter $\eta$. The probability of measuring $\ket{\psi}\ket{0}$ in the first two registers of the resulting state is at least $1-\eta$, in which case the last register collapses to $\frac{1}{\sqrt{\sum_i \pi(i)^2}}\sum_i \pi(i) \ket{i}$. Thus, we can restore $\ket{\psi}$ and sample $i \sim \frac{\pi(i)^2}{\sum_j \pi(i)^2}$ (by measuring the last register) with success probability at least $1-\eta$. The algorithm requires $\bo[\big]{\log(1/\eta)/\sqrt{\lambda}}$ applications of (controlled) $U_{\mathrm{uncomp}}$, $U_{\mathrm{uncomp}}^{\dagger}$, $\id-2\Pi$ and $\id - 2\proj{\psi}$.
\end{proof}

We obtain below the same statement as~\cite{HW20c} for performing nondestructive amplitude estimation. Furthermore, we show that the amplitude estimate is zero with high probability when the amplitude is sufficiently small (as is the case for the regular amplitude estimation~\cite{BHMT02j}).

\begin{proposition}[\sc Nondestructive amplitude estimation]
  \label{Prop:ndAE}
  Let $\ket{\psi}$ be a quantum state and~$\Pi$ be a projection operator with $p = \norm{\Pi \ket{\psi}}^2$. Given $t \geq 1$ and $\eta \in (0,1/2)$, the \emph{nondestructive amplitude estimation} algorithm $\ndae(\ket{\psi},\Pi,t,\eta)$ outputs an estimate $\td{p} \in (0,1)$ such that
    \[\abs*{\td{p} - p} < \frac{\sqrt{p(1-p)}}{t} + \frac{1}{t^2}\]
  with probability at least $1-\eta$. Moreover, if $p \leq 1/(4t^2)$ then it outputs $\td{p} = 0$ with probability at least~$1-\eta$. The algorithm needs one copy of $\ket{\psi}$, which is restored at the end of the computation with probability at least~$1-\eta$, and $\bo{t \log(1/\eta)}$ applications of the reflection operators $\id - 2\proj{\psi}$ and~$\id - 2\Pi$.
\end{proposition}

\begin{proof}
  The standard analysis of the amplitude estimation algorithm (Theorems~11 and 12 in \cite{BHMT02j}) shows that, after $t$ steps, it outputs with probability at least $8/\pi^2$ one of two values $\td{p} \in \set{\td{p}_1, \td{p}_2}$ that both satisfy $\abs*{\td{p}_i - p} < 2\pi\frac{\sqrt{p(1-p)}}{t} + \frac{\pi^2}{t^2}$. Moreover, if $p \leq 1/(4t^2)$ then the probability to output $\td{p} = 0$ is at least $8/\pi^2$. Thus, by taking the median of $\bo{\log(1/\eta)}$ independent runs of amplitude estimation, we obtain $\td{p} \in \set{\td{p}_1, \td{p}_2}$, or $\td{p} = 0$ when $p \leq 1/(4t^2)$, with probability at least~$1-\eta/2$. We apply the amplified uncomputation result (Lemma~\ref{Lem:uncomp}) on this algorithm with amplitude lower bound $\lambda = 1/8$ and success parameter $\eta/2$. If the algorithm succeeds, the probability to sample $\td{p} \in \set{\td{p}_1, \td{p}_2}$ is at least $\frac{2(1/2-\eta/4)^2}{2(1/2-\eta/4)^2 + (\eta/2)^2} \geq 1-\eta/2$, and when $p \leq 1/(4t^2)$ the probability of $\td{p} = 0$ is at least $\frac{(1-\eta/2)^2}{(1-\eta/2)^2 + (\eta/2)^2} \geq 1-\eta/2$.
\end{proof}

%% file: Sources/volume_estimation.tex
In this appendix, we describe a quantum algorithm for estimating the volume of convex bodies that provides an improvement over the algorithm given in \cite{CCH+19p}. The problem statement is as follows. Let $K \subseteq \R^d$ be a convex body, and let $R > 0$ such that $B(1) \subseteq K \subseteq B(R)$, where $B(r) = \set{x \in \R^d : \norm{x}_2 \leq r}$ is the ball of radius~$r$ centered at~$0$. We wish to estimate the volume $\Vol(K)$ of $K$ up to multiplicative precision $\eps$, i.e., to output $\td{V}$ such that $(1-\eps)\Vol(K) \leq \td{V} \leq (1+\eps)\Vol(K)$. We assume to have access to the convex body by means of a membership oracle~$O_K$, i.e., for any point $x \in \R^d$, we can query whether $x \in K$. For the purpose of this paper, we are interested in the query complexity of this problem, i.e., we wish to minimize the number of queries made to this membership oracle. We note here that previous work also consider time-efficient implementations of algorithms that solve this problem~\cite{CCH+19p}. We expect that similar techniques could also be used to implement our algorithm time-efficiently, but we leave this for future work.

First, observe that if the precision $\eps$ is smaller than $\eps \leq (3/4)^d$ then any polylogarithmic overhead in $1/\eps$ is polynomial in the dimension $d$. Thus, in this regime we can run a simple adaption of approximate counting on a suitably dense discretization of~$B(R)$ -- this will run with~$O(1/\eps)$ queries, which when adding polylogarithmic overhead in already achieves the desired $\wbo{d^3 + d^{2.25}/\eps}$ query complexity. Thus, without loss of generality, we can focus on the regime in which $\eps > (3/4)^d$.

Previous works make use of the \textit{pencil construction}, introduced in \cite{LV06j}, where the idea is to define a new convex body $K' \in \R^{d+1}$ with one extra dimension, as
  \[K' = \set*{\bx = (x_0,x) \in \R^{d+1} : x \in K \land x_0 \in [0,2R] \land \norm{x}_2 \leq x_0}.\]
The algorithm now consists of two steps. First, the volume of $K'$ is estimated  up to multiplicative precision $\eps/2$ using the Markov Chain Monte Carlo framework, with the partition function defined as
  \[Z(\beta) = \int_{K'} e^{-\beta x_0} \;\mathrm{d}\bx.\]
For sufficiently large values of $\beta$, the value of the partition function is almost completely determined by the tip of the pencil, whose shape is known to us a priori. On the other hand, for sufficiently small values of $\beta$, the partition function essentially captures the volume of $K'$. These claims are made more precise in \cite{CCH+19p} and the references therein, resulting in Algorithm~4 in said paper which uses a cooling schedule of length $m = \widetilde{\Theta}(\sqrt{d})$. We can speed up this part by using our unbiased product estimator (Theorem~\ref{Thm:prodEst}) in Step~2 of said algorithm. This reduces the number of steps of the quantum walk by a factor of $\sqrt{m} = \widetilde{\Theta}(d^{1/4})$. Hence, we can estimate the volume of~$K'$ with $\wbo{d^3 + d^{2.25}/\eps}$ calls to the membership oracle.

The second step relates the volumes of $K$ and $K'$. It relies on the observation that $R\Vol(K) \leq \Vol(K') \leq 2R\Vol(K)$ since $[R,2R] \times K \subseteq K' \subseteq [0,2R] \times K$. The idea is then to use \emph{rejection sampling} to obtain an $\eps/2$-precise multiplicative estimate of the ratio between $\Vol(K')$ and~$\Vol(K)$. The procedure outlined in \cite[Page~29]{CCH+19p} prepares approximately uniform samples from~$K$ in $\wbo{d^{2.5}}$ calls to the membership oracle. These samples can be trivially converted to samples from $[0,2R] \times K$ by sampling uniformly from the interval $[0,2R]$ and adding the result as an extra dimension. Classically, one could now take $O(1/\eps^2)$ uniform samples from $[0,2R] \times K$, and count the fraction that is contained in $K'$. This yields an $\eps$-precise multiplicative estimate of the ratio between $\Vol(K)$ and $\Vol(K')$. Quantum approximate counting speeds up this step and only requires $O(1/\eps)$ quantum samples, yielding a total of $\wbo{d^{2.5}/\eps}$ queries to the membership oracle in this step of the algorithm.

We claim that the second step can be done with $\wbo{d^2/\eps}$ instead of $\wbo{d^{2.5}/\eps}$ queries, essentially because sampling from~$K$ requires in fact only $\wbo{d^2}$ quantum queries. For simplicity in the proof, we describe a slightly different rejection sampling strategy achieving this complexity. Instead of sampling from $[0,2R] \times K$ and checking whether the sample is contained in~$K'$, we sample from~$K'$ and check whether the resulting sample is contained in $[R,2R] \times K$. Using the same analysis as above, we obtain an $\eps/2$-precise multiplicative estimate of the ratio between~$\Vol(K)$ and~$\Vol(K')$, using~$O(1/\eps)$ quantum samples from~$K'$. Sampling approximately uniformly from~$K'$ can be achieved with the simulated annealing schedule from \cite[Algorithm~3]{CCH+19p} that requires~$\wbo{d^2}$ calls to the membership oracle. It remains to check that the obtained samples from~$K'$ are sufficiently close to uniform, which they are if we marginally increase the length of the cooling schedule:

Let $\eps_1 > 0$ be fixed later, and let $\beta' \leq \eps_1/(2R)$. We show that the Gibbs state at this inverse temperature, denoted by $\ket{\pi'} = \ket{\pi'_{\beta'}}$, is close to the uniform distribution over $K'$, denoted by~$\ket{\pi}$. To that end, recall that $\Vol(K') \geq Z(\beta')$, and so
\begin{align*}
	\Vol(K')(1 - \ip{\pi}{\pi'}) &= \int_{K'} \left(1 - \sqrt{\frac{\Vol(K')}{Z(\beta')}}e^{-\beta' x_0}\right) \;\mathrm{d}\bx \leq \int_{K'} \left(1 - e^{-\beta'x_0}\right) \;\mathrm{d}\bx \\
	&\leq \int_{K'} \beta'x_0 \;\mathrm{d}\bx \leq 2\beta'R\Vol(K') \leq \eps_1\Vol(K'),
\end{align*}
which implies that $\norm{\ket{\pi} - \ket{\pi'}} = 2\sqrt{1-\ip{\pi}{\pi'}} \leq \sqrt{\eps_1}$.

Next, we can prepare $\ket{\pi'}$ by choosing $m = \lceil\sqrt{d}\log(4dR/\eps_1)\rceil$ in \cite[Algorithm~3]{CCH+19p}, where $m$ denotes the length of the cooling schedule. The final inverse temperature then satisfies $\beta' \leq 2d(1-1/\sqrt{d})^{\sqrt{d}\log(4dR/\varepsilon_1)} \leq 2de^{-\log(4dR/\eps_1)} = \eps_1/(2R)$. Since in every annealing step we are performing $N$ Gibbs state reflections, where $N = O(1)$, we can implement these reflections with precision $\sqrt{\eps_1}/(mN)$, to ensure that the total error amounting from the imperfections of these Gibbs state reflections amounts to $\sqrt{\eps_1}$ as well. Thus, we end up preparing a uniform quantum sample up to precision $2\sqrt{\eps_1}$, with a total number of $O(m\log(mN/\sqrt{\eps_1})) = \wbo{\sqrt{d}\log(1/\eps_1)}$ reflections through Gibbs states.

Now, with $\eps_1 = 1/A^2(2N')^2 = \Theta(\eps^2)$, where $N' = \Theta(1/\eps)$ is the number of calls the state-preparation unitary in the approximate counting algorithm, we obtain that the total accumulated error throughout the whole procedure is at most $1/A$. With $A$ a large enough constant, we ensure that the error probability of this step is negligible. Thus, we obtain that the total number of Gibbs state reflections in this second part of the procedure becomes $\wbo{N'\sqrt{d}\log(1/\eps_1)} = \wbo{\sqrt{d}/\eps}$. Since every reflection through a Gibbs state can be performed with $\wbo{d^{1.5}\mathrm{polylog}(1/\eps)}$ membership queries, the total number of calls to $O_K$ becomes $\wbo{d^2/\eps}$. Thus, this second step is less costly than the first step, and the resulting total query complexity becomes $\wbo{d^3 + d^{2.25}/\eps}$.